\documentclass[aps,showpacs,pra,twocolumn]{revtex4-1}

\usepackage[normalem]{ulem}
\usepackage{exscale}
\usepackage{graphicx}
\usepackage{amsmath}
\usepackage{latexsym}
\usepackage{amsfonts}
\usepackage{amssymb}
\usepackage{times}
\usepackage[T1]{fontenc}
\usepackage{lipsum}
\usepackage{amsthm}

\usepackage{caption}
\usepackage{subcaption}

\usepackage{epsfig,pstricks}

\newtheorem{theorem}{Theorem}
\newtheorem{lemma}{Lemma}
\newtheorem{corollary}{Corollary}

\begin{document}

\title{Constructing entanglement measures for fermions}

\author{Markus Johansson$^{1}$ and Zahra Raissi$^{1,2}$} 
\affiliation{$^{1}$ICFO-Institut de Ciencies Fotoniques, The Barcelona Institute of Science and Technology,
08860 Castelldefels (Barcelona), Spain\nonumber\\
$^{2}$Department of Physics, Sharif University of Technology,
Tehran, P.O. Box 111555-9161, Iran}

\date{\today}

\begin{abstract}
In this paper we describe a method for finding polynomial invariants under Stochastic Local Operations and Classical Communication (SLOCC), for a system of delocalized fermions shared between different parties, with global particle number conservation as the only constraint. These invariants can be used to construct entanglement measures for different types of entanglement in such a system. It is shown that the invariants, and the measures constructed from them, take a nonzero value only if the state of the system allows for the observation of Bell-nonlocal correlations. Invariants of this kind are constructed for systems of two and three spin-$1/2$ fermions and examples of maximally entangled states are given that illustrate the different types of entanglement distinguished by the invariants.  
A general condition for the existence of SLOCC invariants and their associated measures is given as a relation between the number of fermions, their spin, and the number of spatial modes of the system. In addition, the effect of further constraints on the system, including the localization of a subset of the fermions, is discussed.
 Finally, a hybrid Ising-Hubbard Hamiltonian is constructed for which the groundstate of a three site chain exhibits a high degree of entanglement at the transition between a regime dominated by on-site interaction and a regime dominated by Ising-interaction. This entanglement is well described by a measure constructed by the introduced method.
\end{abstract}


\maketitle

\section{Introduction}

The phenomenon of entanglement in quantum mechanics is in its essence the presence of action at a distance \cite{epr,Bell}. The operational verification of entanglement thus consist of performing a task that is impossible without nonlocal causation such as for example quantum teleportation \cite{bennett1} or producing Bell-nonlocal correlations \cite{Bell}.

In this paper we study the entanglement in a system of indistinguishable fermions where measurements on the system are performed by spatially separated parties, i.e., experimenters. We assume that the number of fermions is conserved, i.e., that there is a global particle number superlselection rule. Furthermore, we assume that the spatial positions of the fermions are not well defined, i.e., that they are delocalized.

If there is no restriction on the local operations that a party can perform, the local degrees of freedom behave effectively as the degrees of freedom of a localized distinguishable particle.
Therefore, the entanglement properties can be qualitatively different, compared to the case of distinguishable particles, only if the allowed local operations are restricted.
Global particle number conservation can impose such a restriction, but only if at least some of the particles are delocalized, so that the local particle numbers are not fixed. For this reason we consider delocalization of the fermions together with the global particle number conservation.
Another possible restriction on the local operations is particle number parity conservation, which has been considered in Ref. \cite{banuls}.

The entanglement between distinguishable particles without superselection rules has been extensively studied and characterized in a variety of ways for small numbers of particles. See e.g. Refs.  \cite{mahler,wootters,popescu2,coffman,dur,carteret,sudbery,acin,acin2,luquethibon1,verstraete2002,
Osterlohsiewert2005,gour,kraus}. 
Several tools have been developed to describe the entanglement in this case, such as the subsystem entropy \cite{vonneumann,horodecki} and polynomial invariants under local operations \cite{linden,grassl,wootters,popescu2,coffman,sudbery,luquethibon1,Osterlohsiewert2005}. For pure states of such a system the description by these tools has a clear interpretation in terms of the ability to demonstrate nonlocal causation. If the subsystem entropy or the value of a polynomial invariant is nonzero, there exist some local operations that demonstrate nonlocal causation \cite{gisin,popescu3}.

The entanglement of indistinguishable particles has been less explored and the tools to describe it have been less developed. In this case, using the tools created for the case of distinguishable particles can lead to a description that does not carry the same physical meaning.
For example, in a system where the set of allowed operations on the local Hilbert space is restricted, a nonzero entropy of a subsystem or a nonzero
value of an invariant under local operations does not necessarily imply that nonlocal causation can be demonstrated. 

This work describes the group of Stochastic Local Operations and Classical Communication (SLOCC) \cite{bennett} for a system of delocalized indistinguishable spin-$\frac{1}{2}$ fermions with a global particle number superselection rule, and gives a method for constructing polynomial invariants under its action. 
Concrete examples of such invariants are given for systems of two and three spin-$\frac{1}{2}$ fermions with access to a number of spatial modes. 

Furthermore we show that if particle number conservation is the only constraint of the system, a nonzero value of an SLOCC invariant implies that there exist one or more bipartitions of the system for which Bell-nonlocal correlations can be observed. This result extends to fermions with arbitrary spin. We also give a necessary and sufficient condition for the existence of SLOCC invariants as a relation between the number of fermions, their spin, and the number of spatial modes.

For the case of two spin-$\frac{1}{2}$ fermions we compare the constructed SLOCC invariant, and the entanglement measure that is obtained from it, to the subsystem entropy and the fermionic concurrence \cite{mcdonalds}, two other tools that have been used to characterize entanglement and correlations in fermionic systems. It is shown that neither the subsystem entropy nor the fermionic concurrence have a clear relation to Bell-nonlocality.

Briefly we also consider a few cases with additional constraints, such as fixing the spatial location of a subset of the fermions or imposing strong repulsive or attractive interaction.
For the case of strongly attractive interaction the relation between SLOCC invariants and Bell-nonlocality that exists in the unconstrained case is no longer valid.

Finally, as an example we consider a hybrid Ising-Hubbard Hamiltonian where the groundstate entanglement is well described by a measure constructed from a SLOCC invariant. The peak of this groundstate entanglement occurs at a transition between two qualitatively different physical regimes, one dominated by on-site interaction and one dominated by nearest neighbour Ising-interaction.

The outline of the paper is the following.
In Sec. \ref{sec2} the methods for characterizing entanglement and Bell-nonlocality in a system with superselection rules are described. The focus is on the particle conservation superselection rule, and it is described how to construct the group of SLOCC. Section \ref{sec3} gives the polynomial SLOCC invariants for two and three delocalized spin-$\frac{1}{2}$ fermions and discusses their physical meaning. In addition to this two conditions for the existence of SLOCC invariants for arbitrary numbers of fermions and arbitrary spin are given. Section \ref{sec4} contains a brief discussion of SLOCC invariants and measures for cases with additional constraints and a few examples are given. 
In Section \ref{sec5} the Ising-Hubbard Hamiltonian with a groundstate that exhibits entanglement between three fermions at a transition between different physical regimes is described. The paper ends with the conclusions.

\section{Entanglement and Bell-nonlocality of fermions}
\label{sec2}

The setting considered in this paper is a system containing a number of indistinguishable fermions. By fermions we mean particles that satisfy the Pauli exclusion principle, i.e., any two particles with the same internal state cannot exist in the same spatial location. This constraint is equivalent to the requirement that the state vector acquires a minus sign when two particles are interchanged. In particular we focus on the case where each fermion has two internal degrees of freedom. We therefore use the example where the internal degree of freedom is spin and we denote the internal states by $|\!\uparrow\rangle$ and $|\!\downarrow\rangle$, but the results of the paper are independent of the physical nature of the two internal degrees of freedom. Beyond this we assume that the global particle number is conserved.

We also assume that the fermions have access to a finite number of spatial modes and that at least some of the fermions do not have well defined locations, i.e., that they are delocalized. To each spatial mode we associate a party, i.e., an experimenter, capable of performing measurements on any fermion or fermions that occupy its spatial mode. For some purposes we also consider parties with access to more than one spatial mode.

When the locations of the fermions are not well defined it is not operationally meaningful to consider entanglement between degrees of freedom of the individual fermions. Therefore we instead consider entanglement between the degrees of freedom accessible to the different parties \cite{knill,zanardi2,zanardi,fisher,zanardi3,barnumviola}, i.e., entanglement between the modes of the system \cite{zanardi,fisher}. The different parts of the system that can be entangled are therefore defined by the disjoint subsets of operations that the different parties can perform \cite{knill,zanardi2,zanardi3,barnumviola}. In our case what restricts the allowed operations of a party to one of these subsets is the spatial separation of the parties, but more general cases have been considered in Ref. \cite{barnumviola}. 
Entanglement between modes in a system of fermions has been considered previously in Refs. \cite{zanardi,johannesson,banuls} for the case with a particle number parity superselection rule.

Correlations between fermions have been studied also in the terms of conserved quantities in systems of non-interacting fermions \cite{shliemancirackus,mcdonalds,vrana} in particular the Slater rank \cite{shliemancirackus} and its generalizations \cite{vrana}.
The scenario where these quantities are physically relevant is operationally different from the scenario considered here. In particular, it does not involve spatially separated parties performing local operations, but instead considers a system under the restricted set of global operations that do not cause interaction between the fermions.

\subsection{Entanglement and Bell-nonlocality in a system with a superselection rule}

\label{secoo}

The concept of entanglement presupposes a division of the system into parts corresponding to different parties. If we consider the case of three parties A, B, and C we can denote a tri-partitioning of the system by A|B|C. For a partitioning of this type it is assumed that a given party, e.g. party A, can act only on a subset of the degrees of freedom of the system that are locally accessible to A. These degrees of freedom are described by the local Hilbert space $H_A$. For a different partitioning of the system, e.g. a bi-partitioning A|BC, it is assumed that the degrees of freedom of both mode B and C are accessible to a party who can perform operations on the Hilbert space $H_B\otimes H_C$. The full Hilbert space of the shared system is described by the tensor product $H=H_A\otimes H_{B}\otimes H_{C}$ of the individual local Hilbert spaces.
In the case of distinguishable particles and without restrictions on the local operations, a state, represented by a density matrix $\rho$, is entangled with respect to a partitioning if it cannot be expressed as a {\it separable} state. For example, a three-partite state is separable if it can be expressed as

\begin{eqnarray}\label{sep}\rho=\sum_{k}p_k\rho_A^k\otimes\rho_B^k\otimes\rho_C^k,\end{eqnarray} 
where $\rho_X^k$ is a density matrix on $H_{X}$, $p_k>0$, and $\sum_k p_k=1$.
It is clear that the notion of separability depends on the partitioning of the system.
A state that is non-separable for the partition A|B|C may still be separable for the partition A|BC. However, a state which is separable for the most fine grained partitioning, in this case A|B|C, is separable for any more coarse grained partitioning, such as the bi-partitioning A|BC.

In a system with a conserved quantity such as a charge, a particle number, or a parity, the set of physically allowed operations is restricted to those that do not change the given quantity. Such a restriction is called a {\it super-selection rule} and imposes a restriction on the local operations that the parties can perform. See e.g. Ref. \cite{moretti} for an introduction to superselection rules.

For example, when the total value of the conserved quantity is the sum of local values, as is the case with the global particle number, the allowed local operations must conserve these local values. In this case it is possible to decompose the Hilbert space $H$ of a system with $n$ spatial modes as a direct sum $H=\oplus_{q_1,q_2,\dots,q_n} H_{q_1,q_2,\dots,q_n}$ where $H_{q_1,q_2,\dots,q_n}$ is the Hilbert space of the states with local conserved quantities $q_1,q_2,\dots,q_n$ of the $n$ spatial modes. The physically allowed local operations cannot transform between the different sectors defined by the different $H_{q_1,q_2,\dots,q_n}$. In particular, a local operation $g$ on the part of the system with local conserved quantity $q_1$ can not change the value of $q_1$. Any such allowed local operation can therefore be given a block-diagonal matrix form where each block is a transformation on a subspace of the local Hilbert space corresponding to a given possible value of $q_1$, i.e.,

\begin{eqnarray}\label{block}g=\left( \begin{array}{cccccc}
\! {\bf B}_{q_1=r} & \! 0 & 0 & \cdots & 0 & 0\\
\! 0 & {\bf B}_{q_1=r-1} & 0 & \cdots & 0 & 0\\
\! 0 & 0 &  \ddots &  & \vdots & \vdots \\
\! \vdots & \vdots & & \ddots &  0 & 0\\
\! 0 &  0   & \cdots & 0 & {\bf B}_{q_1=1} & 0\\
\! 0 & 0  & \cdots & 0 & 0 & {\bf B}_{q_1=0}\end{array} \!\right),
\end{eqnarray}
where ${\bf B}_{q_1=i}$ is a block containing an operation that preserves the value $q_1=i$ of the local conserved quantity, and $r$ is the maximum possible value of $q_1$. 
This restriction to block diagonal form of the local operations implies that
for two vectors $|\psi\rangle$ and $|\phi\rangle$ that belong to different sectors, i.e., that have different values of the local conserved quantities,
it holds that $\langle\psi|Q|\phi\rangle=0$ for any local operation $Q$.
Therefore, if a state is a coherent superposition $\frac{1}{\sqrt{2}}(|\psi\rangle+|\phi\rangle)$ of two vectors $|\psi\rangle$ and $|\phi\rangle$ it cannot be distinguished from the incoherent mixture
$\frac{1}{2}|\psi\rangle\langle\psi|+\frac{1}{2}|\phi\rangle\langle\phi|$ by any local operations. 

A super-selection rule of this type therefore limits the ability of the parties to distinguish different states. In particular it limits the ability to distinguish non-separable from separable states. More precisely, for a non-separable state $\rho$ there may exist a separable state $\rho_s$ such that all physically allowed observables by parties A, B, and C have the same expectation values. Therefore, we say that $\rho$ is indistinguishable from the separable state $\rho_s$ with respect to the given partitioning A|B|C if for any local measurements $Q_{A},Q_{B},$ and $Q_{C}$ it holds that

\begin{eqnarray}
Tr(Q_{A}\times Q_{B}\times Q_{C} \rho)=Tr(Q_{A}\times Q_{B}\times Q_{C} \rho_s).
\end{eqnarray}
See also Ref. \cite{banuls} for a discussion of separability in fermionic systems with a particle number parity super-selection rule.

A state that is indistinguishable from a separable state for a fine grained partitioning, e.g. A|B|C, may still have entanglement that is detectable given a more coarse partitioning, e.g. A|BC. This is in contrast to the case without super-selection rules where the entanglement of the system can always be fully detected and characterized using the most fine grained partitioning. Thus, in general, the task of operationally characterizing the entanglement in a system with a super-selection rule requires that correlations across all partitionings are considered.

If there is no super-selection rule, parties that share a pure entangled state can always find local measurements   
with outcome correlations that are incompatible with a locally causal description of reality \cite{gisin,popescu3}. Correlations of this kind are called Bell-nonlocal \cite{Bell} and we will refer to states for which such correlations can be observed as Bell-nonlocal states.  
If there is a super-selection rule on the other hand a pure non-separable state may be indistinguishable from a separable state, and in this case all correlations are Bell-local. This is in contrast to the case without super-selection rules where states of this kind do not exist.

An example of a system with four spatial modes and two particles where Bell nonlocal correlations cannot be observed for a quadru-partition, but are present for a bipartition is the electronic two-particle Hanbury Brown-Twiss interferometer \cite{yurke,samuel,neder}. There, two electrons are emitted from independent sources and the path of each electron is split into two paths, i.e., two spatial modes. If the paths are A and B for one electron and C and D for the other, Bell-nonlocality can be observed across AC|BD and AD|BC by means of combining paths in each part of the partition, e.g., A and C, and B and D, respectively, by means of a beam splitter. But nonlocal correlations cannot be observed for AB|CD or A|B|C|D.

\subsection{Representation of three-fermion states}

States of fermions can be described in terms of creation and annihilation operators acting on the unpopulated vacuum state. The main focus in this paper is on a systems with three spin-$\frac{1}{2}$ fermions in three spatial modes under the control of three parties A, B, and C. We therefore give the explicit construction of the global and local Hilbert spaces for this scenario.
We denote the creation and annihilation operators of a fermion with spin up in the spatial mode of party A by $a^{\dagger}_{\uparrow}$ and $a_{\uparrow}$, respectively. In the same way let $a^{\dagger}_{\downarrow}$ and $a_{\downarrow}$ denote the same operators for a fermion with spin down in the spatial mode of A and let the operators $b^{\dagger}_{\uparrow},b_{\uparrow},b^{\dagger}_{\downarrow},b_{\downarrow}$ and $c^{\dagger}_{\uparrow},c_{\uparrow},c^{\dagger}_{\downarrow},c_{\downarrow}$ have the same meaning for parties B and C, respectively.

The application of operators $c^{\dagger}_{\downarrow}a^{\dagger}_{\uparrow}$ to the unpopulated spatial modes creates a state where a fermion with spin up is with party A and a fermion with spin down is with party C. However, applying the operators in the reverse order $a^{\dagger}_{\uparrow}c^{\dagger}_{\downarrow}$ creates a state with the same population of the modes.

These two states are related by an exchange of the two fermions. The antisymmetry under particle exchange implies that the two states are related by a phase-factor $-1$. This is encoded in the anti-commutation relations for the creation and annihilation operators

\begin{eqnarray}
q_ir_j+r_jq_i=0,\phantom{ooo}q^{\dagger}_ir^{\dagger}_j+r^{\dagger}_jq^{\dagger}_i=0,\nonumber\\ q^{\dagger}_ir_j+r_jq^{\dagger}_i=\delta_{qr}\delta_{ij},
\end{eqnarray}
where $q,r\in\{a,b,c\}$ and $i,j\in\{\uparrow,\downarrow\}$. In particular, $q^{\dagger}_iq^{\dagger}_i=0$ which encodes the Pauli exclusion principle that prevents any states with more than one fermion with a given internal state in the same spatial mode.
 
The fermionic exchange antisymmetry leads to a sign ambiguity in the definition of the Hilbert space vectors. This ambiguity can be resolved by using an operator ordering (See e.g. Ref. \cite{kk}).
We therefore use the convention that operators are ordered as $a_{\uparrow},a_{\downarrow},b_{\uparrow},b_{\downarrow},c_{\uparrow},c_{\downarrow}$, meaning an operator in the list is always applied before any other operator to the left of it when the basis vectors are created. Note that we have chosen an ordering where operators that act on the same spatial mode are consecutive. Such a choice naturally captures the spatial separation of A, B, and C and allows for a straightforward matrix representation of local operations.

More precisely, we define the basis vectors of the local Hilbert space of A as
\begin{eqnarray}
a_i^{\dagger}|0\rangle=|i\rangle,\phantom{ooo}
a_{\uparrow}^{\dagger}a_{\downarrow}^{\dagger}|0\rangle=|\diamondsuit\rangle,
\end{eqnarray}
where $i\in\{\uparrow,\downarrow\}$, $0$ denotes an unoccupied spatial mode and $\diamondsuit$ denotes the presence of two particles in a spatial mode.
Completely analogously we define the local Hilbert spaces of B and C by replacing $a_i^{\dagger}$ with $b_i^{\dagger}$ and $c_i^{\dagger}$, respectively. 
The basis vectors of the two party Hilbert space of AB are defined as

\begin{align}
a_i^{\dagger}|00\rangle&=|i0\rangle, & b_j^{\dagger}|00\rangle&=|0j\rangle\nonumber\\
& & a_i^{\dagger}b_j^{\dagger}|00\rangle&=|ij\rangle\nonumber\\
a_{\uparrow}^{\dagger}a_{\downarrow}^{\dagger}|00\rangle&=|\diamondsuit 0\rangle, & 
b_{\uparrow}^{\dagger}b_{\downarrow}^{\dagger}|00\rangle&=|0\diamondsuit\rangle\nonumber\\
a_{\uparrow}^{\dagger}a_{\downarrow}^{\dagger}b_{j}^{\dagger}|00\rangle&=|\diamondsuit j\rangle, & 
a_{i}^{\dagger}b_{\uparrow}^{\dagger}b_{\downarrow}^{\dagger}|00\rangle&=|i\diamondsuit\rangle,
\end{align}
where $i,j\in\{\uparrow,\downarrow\}$. The basis vectors for the Hilbert space of BC are defined as
\begin{align}
b_j^{\dagger}|00\rangle&=|j0\rangle, & c_k^{\dagger}|00\rangle&=|0k\rangle,\nonumber\\
& & b_j^{\dagger}c_k^{\dagger}|00\rangle&=|jk\rangle,\nonumber\\
b_{\uparrow}^{\dagger}b_{\downarrow}^{\dagger}|00\rangle&=|\diamondsuit 0\rangle, &
c_{\uparrow}^{\dagger}c_{\downarrow}^{\dagger}|00\rangle&=|0\diamondsuit\rangle,\nonumber\\
b_{\uparrow}^{\dagger}b_{\downarrow}^{\dagger}c_{k}^{\dagger}|00\rangle&=|\diamondsuit k\rangle, &
b_{j}^{\dagger}c_{\uparrow}^{\dagger}c_{\downarrow}^{\dagger}|00\rangle&=|j\diamondsuit\rangle,
\end{align}
where $j,k\in\{\uparrow,\downarrow\}$, and the basis vectors for the Hilbert space of AC are 
\begin{align}
a_i^{\dagger}|00\rangle&=|i0\rangle, & c_k^{\dagger}|00\rangle&=|0k\rangle,\nonumber\\
& & a_i^{\dagger}c_k^{\dagger}|00\rangle&=|ik\rangle,\nonumber\\
a_{\uparrow}^{\dagger}a_{\downarrow}^{\dagger}|00\rangle&=|\diamondsuit 0\rangle, &
c_{\uparrow}^{\dagger}c_{\downarrow}^{\dagger}|00\rangle&=|0\diamondsuit\rangle,\nonumber\\
a_{\uparrow}^{\dagger}a_{\downarrow}^{\dagger}c_{k}^{\dagger}|00\rangle&=|\diamondsuit k\rangle, &
a_{i}^{\dagger}c_{\uparrow}^{\dagger}c_{\downarrow}^{\dagger}|00\rangle&=|i\diamondsuit\rangle,
\end{align}
where $i,k\in\{\uparrow,\downarrow\}$. Finally, the basis vectors for the full Hilbert space are defined as

\begin{align}\label{kloo}
& & a_i^{\dagger}b_j^{\dagger}c_k^{\dagger}|000\rangle&=|ijk\rangle,\nonumber\\
a_{\uparrow}^{\dagger}a_{\downarrow}^{\dagger}c_k^{\dagger}|000\rangle&=|\diamondsuit 0 k\rangle, &
a_{\uparrow}^{\dagger}a_{\downarrow}^{\dagger}b_j^{\dagger}|000\rangle&=|\diamondsuit j 0\rangle,\nonumber\\
b_{\uparrow}^{\dagger}b_{\downarrow}^{\dagger}c_k^{\dagger}|000\rangle&=|0 \diamondsuit  k\rangle, &
a_i^{\dagger}b_{\uparrow}^{\dagger}b_{\downarrow}^{\dagger}|000\rangle&=|i \diamondsuit  0\rangle,\nonumber\\
a_i^{\dagger}c_{\uparrow}^{\dagger}c_{\downarrow}^{\dagger}|000\rangle&=|i 0\diamondsuit  \rangle, &
b_j^{\dagger}c_{\uparrow}^{\dagger}c_{\downarrow}^{\dagger}|000\rangle&=|0 j\diamondsuit  \rangle,
\end{align}
where $i,j,k\in\{\uparrow,\downarrow\}$.

If a state $|\psi\rangle$ in the full Hilbert space is described by a polynomial in the creation operators that can be factorized into one polynomial that contains only creation operators on a single spatial mode, e.g. $A$, and another polynomial that contains only creation operators on the other modes we express the state as a product state $|\psi\rangle=|\theta\rangle_{A}|\phi\rangle_{BC}$. In particular, all the basis vectors of the Hilbert space, constructed above, can be expressed as $|i\rangle_A|j\rangle_B|k\rangle_C$ for some $i,j,k\in\{\uparrow,\downarrow\,0,\diamondsuit\}$. Note that all product vectors of this form, with particle number from zero to six, span a 64 dimensional space  where the basis vectors with three fermions in Eq. \ref{kloo} span a 20 dimensional subspace.

Next, we note that the operators that conserve local particle number of a spatial mode, e.g. C, are 
polynomials in the creation and annihilation operators where every monomial has an even number of operators, e.g $c_{\uparrow}^{\dagger}c_{\uparrow}$ and $c_{\downarrow}^{\dagger}c_{\uparrow}^{\dagger}c_{\downarrow}c_{\uparrow}$.
Any such polynomial commutes with $a_{i}^{\dagger},b_{j}^{\dagger},a_{k}$  and $b_l$ for $i,j,k,l\in\{\uparrow,\downarrow\}$, and therefore it commutes with any operator on the other spatial modes. 
Because of this, and given the operator ordering where operators acting on the same spatial mode are consecutive, we can represent the local operations that conserve particle number in a tensor product matrix form. For example, the operator $(a^{\dagger}_{\uparrow}a_{\downarrow}+a^{\dagger}_{\downarrow}a_{\uparrow})(b^{\dagger}_{\uparrow}b_{\uparrow}-b^{\dagger}_{\downarrow}b_{\downarrow})c^{\dagger}_{\uparrow}c^{\dagger}_{\downarrow}c_{\downarrow}c_{\uparrow}$ can be represented as

\begin{eqnarray}\label{bl}\left( \begin{array}{cccc}
\! 0 &\,  1 &\, 0 &\, 0\\
\! 1 &\,  0 &\, 0 &\, 0\\
\! 0 &\,  0 &\, 0 &\, 0\\
\! 0 &\,  0 &\, 0 &\, 0\end{array} \!\right)\!\otimes \!\left( \begin{array}{cccc}
\! 1 & \!\! 0 &\! 0 &\, 0\\
\! 0 & \!\! -1 &\! 0  &\, 0\\
\! 0 &\!\! 0 &\! 0 &\, 0 \\
\! 0 &\!\! 0 &\! 0 &\, 0 \end{array} \!\right)\!\otimes\!\left( \begin{array}{cccc}
\! 0 &\,  0 &\, 0\! &\, 0 \\
\! 0 &\,  0 &\, 0\! &\, 0 \\
\! 0 &\, 0 &\, 0\! &\, 0 \\
\! 0 &\, 0 &\, 0\! &\, 1\end{array} \!\right),
\end{eqnarray}
where each individual matrix acts on the local 4 dimensional Hilbert space and is given 
in the basis $|\!\uparrow\rangle,|\!\downarrow\rangle,|0\rangle,|\diamondsuit\rangle$. The tensor product of the three matrices act on the 64 dimensional space spanned by the product vectors. We can thus see that a representation by an operator ordering where operators acting on the same spatial mode are consecutive effectively gives the Hilbert space a tensor product structure. 
This facilitates the study of entanglement between the local degrees of freedom of the different parties, represented by $|\!\uparrow\rangle,|\!\downarrow\rangle,|0\rangle,|\diamondsuit\rangle$.

Finally, we consider how to construct single-party reduced density matrices. A matrix  $\rho_{C}$ can be called a reduced density matrix for party C if for any physically allowed local operation $Q_C$ on $H_C$ it satisfies $Tr_{ABC}(Q_C\rho_{ABC})=Tr_{C}( Q_C\rho_C)$, where $\rho_{ABC}=\sum_{i,j,k,l,m,n}r_{ikm}r_{jln}^*|ikm\rangle\langle jln|$ is the density matrix of the full state.
The reduced density matrix $\rho_{C}$ is defined as $\rho_{C}=Tr_{AB}(\rho_{ABC})=\sum_{kl}\langle kl|\rho_{ABC}|kl\rangle_{AB}$ for $k,l\in\{\uparrow,\downarrow,0,\diamondsuit\}$. 
Since the particle conserving local operations by C commute with $a_{i}^{\dagger},b_{j}^{\dagger}$ for any $i,j\in\{\uparrow,\downarrow\}$ it follows that $\langle ikn|Q_C|ikm\rangle=\langle n|Q_C|m\rangle$ for any $i,k,m,n\in\{\uparrow,\downarrow\,0,\diamondsuit\}$. Therefore, $\rho_{C}$ is given by $\rho_{C}=\sum_{m,n}r_{mn}|m\rangle\langle n|$ where $r_{mn}=\sum_{ik}{r_{ikm}r_{ikn}^*}$. The reduced density matrices of A and B can be defined analogously.

\subsection{Local operations satisfying the particle conservation superselection rule}

The most general type of local operations that can be performed is the Stochastic Local Operations and Classical Communication (SLOCC) \cite{bennett}. The operational meaning of SLOCC is that each party can entangle their modes with a local ancillary system, and subsequently perform a projective measurement on the ancillary system. This is then followed by postselection based on classical communication of the measurement outcomes. We limit our interest to the SLOCC operations that can be reversed in a probabilistic sense, i.e., the operations $g$ which satisfy that $|\psi\rangle$ can be obtained from $g|\psi\rangle$ by performing another SLOCC operation. It is with respect to the action of these reversible operations that invariants can be defined. 
 
As described in Sect. \ref{secoo}, a global conservation of particle number implies that any local operation made on the system must conserve the local particle number. This implies a restriction of the operations to the block diagonal form given in Eq. \ref{block}. 
The reversible SLOCC on $n$ spatial modes satisfying the particle conservation super-selection rule is mathematically described by a block diagonal subgroup $G_{{\mbox{\tiny\itshape SLOCC}}}$ of ${\mathrm{SL}}^{\times n}$, where SL is the special linear group. For the case of spin-$\frac{1}{2}$ fermions the local action of $G_{{\mbox{\tiny\itshape SLOCC}}}$ is a subgroup of SL(4). Locally on each spatial mode the action of $G_{{\mbox{\tiny\itshape SLOCC}}}$ is described by block diagonal matrices $g$ of the form

\begin{eqnarray}\label{bl}g=\left( \begin{array}{cccc}
\!\! c_{\uparrow\uparrow} & \! c_{\uparrow\downarrow} & 0 & 0\\
\!\! c_{\downarrow\uparrow} & c_{\downarrow\downarrow} & 0 & 0\\
\!\! 0 &\! 0 & c_{00} & 0 \\
\!\! 0 &\! 0 & 0 & c_{\diamondsuit\diamondsuit}\end{array} \!\right),\phantom{uu}det(g)=1,
\end{eqnarray}
where $c_{ij}=\langle i|g|j\rangle$ for the basis $|\!\uparrow\rangle,|\!\downarrow\rangle,|0\rangle, |\diamondsuit\rangle$ of the local Hilbert space. The particle conservation super-selection rule is manifested in a block diagonal structure with three blocks. The $2\times 2$ block corresponds to the action on a single fermion in the mode, one of the $1\times 1$ blocks corresponds to the action on two fermions in the mode and the other $1\times 1$ block corresponds to the action in the absence of a fermion. Note that the particle conservation also implies that every single-party reduced density matrix has the same block-diagonal form as the local operations.

The group $G_{{\mbox{\tiny\itshape SLOCC}}}$ is a Lie-group \cite{hall} of determinant one matrices. Any element in such a group can be expressed as the matrix exponential of a traceless matrix and the structure of the a group can be described in terms of a set of generators. 
The generators of the Lie group are defined as a set of traceless matrices such that any element of the group can be expressed as a matrix exponential of some linear combination of the generators.

For a given spatial mode, the block diagonal structure in Eq.\ref{bl} implies that the local action of the group $G_{{\mbox{\tiny\itshape SLOCC}}}$ is generated by the five Hermitian traceless matrices

\begin{align}\lambda_1=&\left( \begin{array}{cccc}
\! 0 &\,  1 &\, 0 &\, 0\\
\! 1 &\,  0 &\, 0 &\, 0\\
\! 0 &\,  0 &\, 0 &\, 0\\
\! 0 &\,  0 &\, 0 &\, 0\end{array} \!\right), &\lambda_2=\left( \begin{array}{cccc}
\! 0 &\!  -i & 0 &\, 0\\
\! i &\!  0 & 0 &\, 0\\
\! 0 &\! 0 & 0 &\, 0\\
\! 0 &\!  0 & 0 &\, 0\end{array} \!\right),\nonumber\\
\lambda_3=&\left( \begin{array}{cccc}
\! 1 & \!\! 0 &\! 0 &\, 0\\
\! 0 & \!\! -1 &\! 0  &\, 0\\
\! 0 &\!\! 0 &\! 0 &\, 0 \\
\! 0 &\!\! 0 &\! 0 &\, 0 \end{array} \!\right),  &\lambda_8=\left( \begin{array}{cccc}
\! 1 &\,  0 &\! 0\! & 0 \\
\! 0 &\,  1 &\! 0\! & 0 \\
\! 0 &\, 0 &\!\! -2\! & 0\\
\! 0 &\, 0 &\! 0\! & 0\end{array} \!\right),\nonumber\\
\lambda_{15}=&\left( \begin{array}{cccc}
\! 1 &\,  0 &\, 0\! & 0\! \\
\! 0 &\,  1 &\, 0\! & 0\! \\
\! 0 &\, 0 &\, 1\! & 0\! \\
\! 0 &\, 0 &\, 0 &\!\! -3\end{array} \!\right).
\end{align}
For example, this means that in the case of three parties any element of $G_{{\mbox{\tiny\itshape SLOCC}}}$ can be expressed as $e^{\alpha\lambda_1+\beta\lambda_2+\gamma\lambda_3+\delta\lambda_8+\epsilon\lambda_{15}}\times e^{\zeta\lambda_1+\eta\lambda_2+\theta\lambda_3+\kappa\lambda_8+\mu\lambda_{15}}\times e^{\nu\lambda_1+\xi\lambda_2+\tau\lambda_3+\chi\lambda_8+\omega\lambda_{15}}$ for some $\alpha,\dots,\omega \in \mathbb{C}$.
See e.g. Ref. \cite{pff} for a description of the generators of general SU and SL groups.

We can define the subgroup $G_{1,2,3}$ of $G_{{\mbox{\tiny\itshape SLOCC}}}$ as the group with local action generated by $\lambda_1,\lambda_2,$ and $\lambda_3$. This group acts only on the local subspace spanned by $|\!\uparrow\rangle$ and $|\!\downarrow\rangle$ for each spatial mode. Analogously, we can define the subgroup $G_{8,15}$ as the group with local action generated by  $\lambda_8$ and $\lambda_{15}$. The group $G_{8,15}$ contains the only operations with a non-trivial action on  $|0\rangle$ and $|\diamondsuit\rangle$ which are allowed by the super-selection rule. For example, $e^{i\pi/4\lambda_{15}}\times I\times I$ corresponds to a particle exchange of the fermions in the spatial mode of party A, conditioned on the presence of two fermions in this mode, and a $e^{i\pi/4}$ global phase shift. It is clear from the form of the generators that $G_{1,2,3}$ and $G_{8,15}$ commute. Therefore, we can express  $G_{{\mbox{\tiny\itshape SLOCC}}}$ as a product of
$G_{1,2,3}$ and $G_{8,15}$, i.e., $G_{{\mbox{\tiny\itshape SLOCC}}}=G_{1,2,3}G_{8,15}$.

For pure states of three fermions in three spatial modes the structure of $G_{{\mbox{\tiny\itshape SLOCC}}}$ leads to a simple relation between Bell-nonlocality and the ability to distinguish the state from a separable state. Such a state is Bell-nonlocal if and only if it is distinguishable from a separable state. 

\begin{lemma}\label{lemm1}
A pure non-separable state of three delocalized fermions is Bell-local with respect to a given  partitioning if and only if it is indistinguishable from a separable state with respect to the same partitioning.
\end{lemma}
\begin{proof}
The if part of the statement is trivial. We therefore consider the only if part.

For a bi-partitioning, e.g. A|BC, a general state $|\psi\rangle$ can be written as
\begin{eqnarray}|\psi\rangle=r_1|\!\uparrow\rangle|\theta\rangle+r_2|\!\downarrow\rangle|\xi\rangle+r_3|0\rangle|\phi\rangle+r_4|\diamondsuit\rangle|\varphi\rangle,\end{eqnarray} 
where $|\theta\rangle,|\xi\rangle,|\phi\rangle,|\varphi\rangle\in H_{BC}$. Consider the projection of this state onto the subspace where one fermion is localized with A and the two others with BC. This projection is $r_1|\!\uparrow\rangle|\theta\rangle+r_2|\!\downarrow\rangle|\xi\rangle$. All states in the subspace spanned by $|\theta\rangle$ and $|\xi\rangle$ contain two fermions. Therefore, within this subspace the operations by the party with access to the modes B and C are not restricted by the super-selection rule. 
Because of this there exist measurements that exhibit nonlocal outcome correlations provided that the state is entangled as shown in Ref. \cite{gisin}. 
Any outcomes where party A finds $|0\rangle$ or $|\diamondsuit\rangle$ has only local correlations regardless of the state of BC since nonlocal correlations require at least two incompatible measurements performed by A \cite{Bell}. Such measurements cannot be performed on the subspace spanned by $|0\rangle$ or $|\diamondsuit\rangle$ due to the particle number conservation superselection rule.

The state is thus Bell-local if and only if $r_1r_2=0$ or $|\theta\rangle=|\xi\rangle$. Such a state is, up to local unitary operations, of the form $s_1|\!\downarrow\rangle|\varphi\rangle+s_2|0\rangle|\phi\rangle+s_3|\diamondsuit\rangle|\zeta\rangle$ and since the three terms belong to different super selction sectors this state is indistinguishable from a separable state $|s_1|^2|\!\downarrow\rangle\langle\downarrow\!|\otimes|\varphi\rangle\langle\varphi|+|s_2|^2|0\rangle\langle 0|\otimes|\phi\rangle\langle\phi|+|s_3|^2|\diamondsuit\rangle\langle \diamondsuit|\otimes|\zeta\rangle\langle\zeta|$. 

For a partitioning A|B|C we consider the projection of the state onto a subspace where one fermion is with each party. If this projection is an entangled pure state there are measurements with nonlocal outcome statistics \cite{popescu3}. If on the other hand the projection is a product state the outcome statistics is trivially Bell-local.

If instead party A finds its local state to be $|0\rangle$ the state of parties B and C is projected onto the subspace spanned by $|\!\uparrow\!\diamondsuit\rangle$, $|\!\downarrow\!\diamondsuit\rangle$, $|\diamondsuit\!\uparrow\rangle$, and $|\diamondsuit\!\downarrow\rangle$. 
No pair of local operations on B and C can transform between a vector $|i\diamondsuit\rangle$ and a vector $|\diamondsuit j\rangle$ for any $i,j\in\{\uparrow,\downarrow\}$ because of the particle number conservation super-selection rule. This implies that all states within this subspace are indistinguishable from separable states. Furthermore, for a vector of the kind $|i\diamondsuit\rangle$ or $|\diamondsuit j\rangle$, incompatible measurements can only be performed on one of the two spatial modes.
Therefore any correlation between measurement outcomes are local. 
The same argument can be made if we exchange the roles of $|0\rangle$ and $|\diamondsuit\rangle$, and for any permutation of the parties A, B, and C.

\end{proof}

\subsection{Entanglement classification using invariants}
\label{ent}

To characterize the entanglement of a system one can describe it in terms of quantities and parameters specific to a particular physical realization. However, our main interest is realization independent properties, i.e., intrinsic properties of the system.

An often used tool that depends only on intrinsic properties and describes entanglement between two parts of a system is the subsystem entropy \cite{vonneumann,horodecki}. However, since the entropy is a single real valued function it cannot in general capture all the different ways in which the parts of the system can be entangled. The subsystem entropy has been used in the study of fermionic systems in for example Refs. \cite{johannesson,wolf}.
A tool that can characterize all the different types of entanglement is the polynomial entanglement invariants \cite{linden,grassl,wootters,popescu2,coffman,sudbery,luquethibon1,Osterlohsiewert2005}.
Polynomial invariants can be constructed both for the local unitary operations and for SLOCC. Here we focus on the invariants under SLOCC. 
A polynomial invariant of $G_{{\mbox{\tiny\itshape SLOCC}}}$ is defined as a polynomial in the coefficients of the state vector which is invariant under the action of $G_{{\mbox{\tiny\itshape SLOCC}}}$. 
Two states are said to have the same type of entanglement with respect to $G_{{\mbox{\tiny\itshape SLOCC}}}$ if they can be converted into each other by $G_{{\mbox{\tiny\itshape SLOCC}}}$ with a nonzero probability. If this is the case any ratio between two $G_{{\mbox{\tiny\itshape SLOCC}}}$ invariants take the same value for the two states.
In this sense the invariants under $G_{{\mbox{\tiny\itshape SLOCC}}}$ serve as coordinates on the space of $G_{{\mbox{\tiny\itshape SLOCC}}}$ inter-convertibility classes, i.e., different types of entanglement with respect to $G_{{\mbox{\tiny\itshape SLOCC}}}$. However, note that there exist a subset of the inter-covertibility classes that cannot be distinguished by the $G_{{\mbox{\tiny\itshape SLOCC}}}$ invariants and even a subset where all $G_{{\mbox{\tiny\itshape SLOCC}}}$ invariants take the value zero. These classes can only be distinguished by invariants under local unitary operations.
A subset of the invariants with the property that all
invariants can be constructed algebraically from the subset is
called a set of generators \cite{olver}. Any such set of generators thus gives the same description of the entanglement types as the full set of invariants. 
Suitably chosen the generators correspond to physically relevant properties of the system that depend on the entanglement.

Another use for the SLOCC invariants is that any such invariant can be used to construct a function that is non-increasing under SLOCC, a so called {\it entanglement monotone} \cite{vidal,verstraete2003}. It was shown in Ref. \cite{verstraete2003} that if $I$ is a SLOCC invariant of degree $k$, the function $|{I}|^{2/k}$ is an entanglement monotone. To understand the precise meaning of such a non-increasing function under SLOCC we must understand two  properties of SLOCC operations and how they are mathematically represented. First, a SLOCC operation, when represented as a determinant one matrix, does not generally preserve the norm of the state vector. This implies that if the state vector is normalized after the SLOCC operation has been performed, the value of $I$ for the output state is in general not the same as that of the input state. Secondly, a SLOCC operation can in general not be performed with unit probability of success. The outcome state is a statistical mixture of the desired outcome state and some other state(s), and the desired SLOCC operation is realized only after a postselection to remove the undesired outcome states.

If the average of $I$ is taken over all outcome states, including the undesired ones, and these states are properly normalized, this average is typically not the same as the value for $I$ on the input state. However, the average value of the function $|{I}|^{2/k}$ does not increase \cite{verstraete2003}.
A function of this kind that does not increase under SLOCC, an entanglement monotone,  quantifies a property of the system that cannot be created by SLOCC. In a system without superselection rules this directly implies that it quantifies a property of the entanglement since all non-entangled states can be interconverted by SLOCC. In this case the monotone serves as a measure of the particular type of entanglement that is distinguished by the invariant it was constructed from.

Closely related to entanglement measures is the notion of maximally entangled states.
A maximally entangled state is defined as a pure state for which all the single-party reduced density matrices are maximally mixed \cite{mahler}. Such a state thus has the maximal possible single-party subsystem entropy.
If a state $|\psi\rangle$ and a maximally entangled state can be interconverted by SLOCC, there exist at least one entanglement monotone which takes a non-zero value for $|\psi\rangle$. Furthermore, any monotone attains its maximal value, on the SLOCC inter-convertibility class of $|\psi\rangle$, on the set of maximally entangled states that belong to the class \cite{kempfness,verstraete2003}. If there is no maximally entangled state in a given SLOCC inter-convertibility class there still exist at least one entanglement monotone that take nonzero values if the states in the class can be brought arbitrarily close to maximally entangled states by SLOCC.

The results in Theorem 1 and 2 of Ref. \cite{verstraete2003} that relate SLOCC invariants to SLOCC monotone functions and maximally entangled states are formulated in the context of distinguishable particles where the local action of SLOCC operations is described by general SL matrices. Nevertheless, the theorems are valid also in the case considered here with block-diagonal subgroups of SL and block diagonal single-party reduced density matrices. 
However, in the presence of a particle conservation superselection rule it is not a priori clear that
a monotone measures a property of the state that corresponds to non-local causation, since the set of SLOCC operations is restricted. 
We therefore show that in the case of three spin-$\frac{1}{2}$ fermions a nonzero value of an entanglement monotone implies that Bell-nonlocal correlations can be observed across all bi-partitions.

\begin{theorem}\label{t1}
If a state $|\psi\rangle$ of three spin-$\frac{1}{2}$ fermions is Bell-local over a bipartition, there does not exist any entanglement monotone that takes a non-zero value for $|\psi\rangle$ and thus there exist no maximally entangled state 
that $|\psi\rangle$ can be converted to by SLOCC. Moreover, if $|\psi\rangle$ is Bell local any global phase shift of the state vector $|\psi\rangle$ can be induced by a special unitary operation. Thus, if an entanglement monotone takes a non-zero value the state is Bell-nonlocal over every bi-partition.
\end{theorem}
\begin{proof}

As described in the proof of Lemma \ref{lemm1}, if a three spin-$\frac{1}{2}$ fermion state $|\psi\rangle$ is Bell-local over a bipartition A|BC there exist a basis such that
$|\psi\rangle=r_1|\!\uparrow\rangle|\theta\rangle+r_2|0\rangle|\phi\rangle+r_3|\diamondsuit\rangle|\varphi\rangle$. For a state of this form the determinant one matrix $e^{-6\alpha\lambda_3+2\alpha\lambda_8+\alpha\lambda_{15}}\times{I}\times{I}$  scales the state vector by a factor of $|e^{-3\alpha}|$ and changes the phase by $\arg(e^{-3\alpha})$. By making $Re(\alpha)$ sufficiently large we can bring the norm arbitrarily close to zero. This implies that the set of states which are inter-convertible with $|\psi\rangle$ does not contain any maximally entangled state. This follows from the Kempf-Ness theorem \cite{kempfness}, and Theorem 1 of Ref. \cite{verstraete2003}. It also implies that there does not exist any polynomial SLOCC invariant that takes a nonzero value. This follows since no homogeneous polynomial is invariant under a scaling of the state vector.
Thus, if there exist an entanglement monotone which takes a non-zero value for a given state this state cannot have the form $|\psi\rangle=r_1|\!\uparrow\rangle|\theta\rangle+r_2|0\rangle|\phi\rangle+r_3|\diamondsuit\rangle|\varphi\rangle$ for any bipartition and must be Bell-nonlocal over every bipartition.

Moreover, the form of the state $|\psi\rangle$ also implies that
$e^{3i\alpha}e^{-6i\alpha\lambda_3+2i\alpha\lambda_8+i\alpha\lambda_{15}}\times{I}\times{I}$ is an element of the stabilizer group. This means that under the special unitary operation $e^{-6i\alpha\lambda_3+2i\alpha\lambda_8+i\alpha\lambda_{15}}\times{I}\times{I}$ the state accumulates a phase-factor $e^{-3i\alpha}$.

\end{proof}
Thus, for the case of three spin-$\frac{1}{2}$ fermions a monotone constructed from a $G_{{\mbox{\tiny\itshape SLOCC}}}$ invariant serve as a measure of distinguishable entanglement, i.e., as a measure of some property of the system that makes non-local causation possible. In particular the monotone serves as a measure of the type of distinguishable entanglement that is characterized by the invariant it was constructed from.

Note that a SLOCC invariant can be nonzero for a state only if all of its one party reduced density matrices are full rank. If the system does not have states satisfying this property it is impossible to construct SLOCC invariants that take nonzero values, and there are no maximally entangled states. But beyond this requirement we can give a stronger constraint on the existence of SLOCC invariants and maximally entangled states.

\begin{theorem}\label{lut}
There exist entanglement measures and maximally entangled states for a system of spin-$\frac{1}{2}$ fermions only if the number of spatial modes equals the number of fermions.
\end{theorem}
\begin{proof}
Consider $m$ spin-$\frac{1}{2}$ fermions in a system of $n$ spatial modes and the matrix $g$

\begin{eqnarray}g=\left( \begin{array}{cccc}
 1 &\,\,\,\,\,\,\,\, 0 &\,\,\,\, 0 &\!\!\!\!\! 0\\
 0 &\,\,\,\,\,\,\,\, 1 &\,\,\,\, 0 &\!\!\!\!\! 0\\
 0 &\,\,\,\,\,\,\,\, 0 &\,\,\,\, re^{i\phi}&\!\!\!\!\! 0\\
 0 &\,\,\,\,\,\,\,\, 0 &\,\,\,\, 0 &\!\!\!\! r^{-1}e^{-i\phi}\end{array}\!\!\right).\nonumber
\end{eqnarray}
The action of $g^{\times{n}}\in G_{{\mbox{\tiny\itshape SLOCC}}}$ on any basis vector of the Hilbert space depends only on the number of fermions $m$ and the number of spatial modes $n$. In particular, the action is a scaling of the state vector by the factor $r^{(n-m)}e^{(n-m)i\phi}$.  

It follows that for $r\neq{1}$ it is impossible to construct polynomials in the state vector coefficients that are invariant under this action unless $n=m$. 
\end{proof}
Thus, if the number of spin-$\frac{1}{2}$ fermions does not equal the number of spatial modes the entanglement of the system can never be maximal in the sense that the single-party reduced density matrices are proportional to the identity.

\section{The unconstrained case}
\label{sec3}
In this section we consider the case where the only physical constraint on the system of fermions is the global particle number conservation. First we consider systems with two or three spin-$\frac{1}{2}$ fermions and two or three spatial modes, respectively, and SLOCC invariants for these systems are constructed. In addition to this, two general conditions for the existence of SLOCC invariants, which are valid for arbitrary spin and number of particles, are given.

\subsection{Two parties and two spin-$\frac{1}{2}$ fermions}
\label{ferm}

For a system of two spin-$\frac{1}{2}$ fermions and two spatial modes, where each mode is associated to a party that can make local operations, the SLOCC invariants are generated by a single SLOCC invariant $I_0$
\begin{eqnarray}
I_0=(m_{\uparrow\uparrow}m_{\downarrow\downarrow}-m_{\uparrow\downarrow}m_{\downarrow\uparrow})m_{0\diamondsuit}m_{\diamondsuit 0},
\end{eqnarray}
where $m_{ij}$ is the amplitude of $|ij\rangle$ in the state vector.
The invariant is a product of the concurrence polynomial \cite{wootters} which is the only generator of the $G_{1,2,3}$ invariants, and the factor $m_{0\diamondsuit}m_{\diamondsuit 0}$ which is necessary for invariance under $G_{8,15}$.

A measure can be constructed from $I_0$ as

\begin{eqnarray}
|I_0|^{1/2}=|(m_{\uparrow\uparrow}m_{\downarrow\downarrow}-m_{\uparrow\downarrow}m_{\downarrow\uparrow})m_{0\diamondsuit}m_{\diamondsuit 0}|^{1/2}.
\end{eqnarray}

The fact that the SLOCC invariants are generated by a single polynomial implies that there exist only a single maximally entangled state up to local unitary transformations. This state is given by 

\begin{eqnarray}\label{maxi}
\frac{1}{2}(|\!\uparrow\uparrow\rangle+|\!\downarrow\downarrow\rangle+|0\diamondsuit\rangle+|\diamondsuit 0\rangle).
\end{eqnarray}
The measure $|I_0|^{1/2}$ takes its maximum value $1/4$ for this state.

\subsubsection{Comparison with the Slater rank and fermionic concurrence}

Here we briefly compare the description of entanglement using the invariant $I_0$ and its corresponding measure to the Slater rank and fermionic concurrence constructed in Ref. \cite{shliemancirackus} and Ref. \cite{mcdonalds}, respectively.  The context of both these constructions is the description of correlations in a system of non-interacting fermions undergoing global unitary evolution. In such a non-interacting system the Slater rank \cite{shliemancirackus} is a conserved quantity and has been considered as a way to characterize correlations between the fermions. The Slater rank is defined as the minimal number of Hilbert space vectors needed to describe the state of the system, up to global unitary operations that do not cause interaction between the particles. For a system of $n$ spatial modes this group is $U(2n)$ with an action on each of the single particle Hilbert spaces. This scenario is different from the standard entanglement scenario since it does not feature local parties with access to different parts of the system but instead a single party with access to the whole system. See also Ref. \cite{fisher} for a discussion comparing the different characterizations of correlations.

The group of all local unitary operations include some operations that cause interaction and are therefore not allowed in the non-interacting scenario of Ref. \cite{shliemancirackus}. The Slater rank of two fermions is not a conserved quantity under these operations. An example of a local operation that causes interaction is a particle exchange in the spatial mode of party A, represented by $e^{-i\pi/4}e^{i\pi/4\lambda_{15}}\times{I}$. To see this consider the state $|\psi_1\rangle$ with Slater rank 1

\begin{align}
|\psi_1\rangle&=\frac{1}{2}(|\!\downarrow\uparrow\rangle-|\!\uparrow\downarrow\rangle-|\diamondsuit 0\rangle-|0\diamondsuit\rangle)\nonumber\\
&=\frac{1}{2}(a_{\downarrow}^{\dagger}b_{\uparrow}^{\dagger}-a_{\uparrow}^{\dagger}b_{\downarrow}^{\dagger}-a_{\uparrow}^{\dagger}a_{\downarrow}^{\dagger}-b_{\uparrow}^{\dagger}b_{\downarrow}^{\dagger})|00\rangle\nonumber\\
&=\frac{1}{2}(a_{\downarrow}^{\dagger}+b_{\downarrow}^{\dagger})(a_{\uparrow}^{\dagger}+b_{\uparrow}^{\dagger})|00\rangle.
\end{align}
The local unitary operation $e^{-i\pi/4}e^{i\pi/4\lambda_{15}}\times{I}$ transforms $|\psi_1\rangle$ to the state $|\psi_2\rangle$ given by

\begin{align}
|\psi_2\rangle&=\frac{1}{2}(|\!\downarrow\uparrow\rangle-|\!\uparrow\downarrow\rangle+|\diamondsuit 0\rangle-|0\diamondsuit\rangle)
\nonumber\\
&=\frac{1}{2}(a_{\downarrow}^{\dagger}b_{\uparrow}^{\dagger}-a_{\uparrow}^{\dagger}b_{\downarrow}^{\dagger}+a_{\uparrow}^{\dagger}a_{\downarrow}^{\dagger}-b_{\uparrow}^{\dagger}b_{\downarrow}^{\dagger})|00\rangle\nonumber\\
&=\frac{1}{2}[a_{\downarrow}^{\dagger}(b_{\uparrow}^{\dagger}-a_{\uparrow}^{\dagger})+b_{\downarrow}^{\dagger}(a_{\uparrow}^{\dagger}+b_{\uparrow}^{\dagger})]|00\rangle,
\end{align} 
which has Slater rank 2.
Both, $|\psi_1\rangle$ and $|\psi_2\rangle$ are maximally entangled and thus $|I_0|^{1/2}$ takes its maximum value $1/4$ for both states. 

It can also be noted that the Slater rank does not make a difference between states that are
distinguishable and states that are indistinguishable from separable states. The state $|\psi_3\rangle$ is the familiar Bell state while the state $|\psi_4\rangle$ is a superposition of a pair of fermions being either with A or with B and cannot be distinguished from a separable state 
\begin{align}\label{eqo}
|\psi_3\rangle&=\frac{1}{\sqrt{2}}(|\!\uparrow\downarrow\rangle+|\!\downarrow\uparrow\rangle)\nonumber\\
|\psi_4\rangle&=\frac{1}{\sqrt{2}}(|\diamondsuit 0\rangle+|0\diamondsuit\rangle).
\end{align}
Both of the states have Slater rank 2.

In relation to the Slater rank a correlation measure, the fermionic concurrence, defined by

\begin{eqnarray}
|m_{0\diamondsuit}m_{\diamondsuit 0}-m_{\uparrow\uparrow}m_{\downarrow\downarrow}+m_{\uparrow\downarrow}m_{\downarrow\uparrow}|,
\end{eqnarray}
was introduced in Ref. \cite{mcdonalds} and is zero if and only if the Slater rank is 1.
Since the Slater rank is not conserved by local operations neither is the fermionic concurrence. The fermionic concurrence is zero for the state $|\psi_1\rangle$ but equals $1/2$ for $|\psi_2\rangle$.

It is thus clear that the description of correlations in Ref. \cite{shliemancirackus} and the description of entanglement considered in this paper describe qualitatively different properties of the system. Moreover, the Slater rank and the fermionic concurrence are not invariant under local unitary operations and the latter can thus not be used as an entanglement measure.

\subsubsection{Comparison with the subsystem entropy}

For a pure state $|\Psi\rangle$ the subsystem entropy $\mathcal{E}\equiv -Tr[\rho_A \ln(\rho_A)]$ of the reduced density matrix $\rho_A$ of party A quantifies the extent of mixedness of $\rho_A$. In a system without superselection rules this is in direct correspondence to entanglement and Bell-nonlocality. If $\mathcal{E}=0$ party A is unentangled with the rest of the system and the state is Bell local, and if $\mathcal{E}>0$ there is entanglement and Bell-nonlocality. In particular if $\mathcal{E}=\ln(d)$, where $d$ is the dimension of the local Hilbert space, the state is maximally entangled with the rest of the system.  

In the case with the particle conservation super-selection rule the connection with entanglement and Bell-nonlocality is not this simple. For example the two states $|\psi_3\rangle$ and $|\psi_4\rangle$ in Eq. \ref{eqo} both have the subsystem entropy $\ln(2)$. While $|\psi_3\rangle$ is a maximally entangled Bell state of two localized particles, the state $|\psi_4\rangle$ is indistinguishable from a separable state.
Nevertheless, the subsystem entropy reaches its maximal value $\ln(4)$ for the maximally entangled state in Eq. \ref{maxi}. Moreover, if the state is Bell local it follows that $\mathcal{E}\leq\ln(3)$. This can be seen from the proof of Lemma \ref{lemm1} since a Bell-local state can always be written $s_1|\!\downarrow\rangle|\varphi\rangle+s_2|0\rangle|\phi\rangle+s_3|\diamondsuit\rangle|\zeta\rangle$ up to local unitary transformations. The reduced density matrix is then of rank 3 or less. Thus, for a pure state a subsystem entropy above $\ln(3)$ guarantees Bell-nonlocality but no conclusions can be made for a lower but nonzero value.

\subsection{Three spin-$\frac{1}{2}$ fermions and three parties}
For three spin-$\frac{1}{2}$ fermions shared between three parties there is at least 7 independent SLOCC invariants. As opposed to the case of three distinguishable spin-$\frac{1}{2}$ particles where the number of SLOCC inter-convertibility classes is finite, in this case there is an infinite number of such classes parametrized by at least these 7 invariants.
 
The method used here to construct the $G_{{\mbox{\tiny\itshape SLOCC}}}$ invariants is to find the algebra of polynomial $G_{1,2,3}$-invariants and then select the subalgebra of $G_{8,15}$ invariants. Since $G_{1,2,3}$ and $G_{8,15}$ commute this subalgebra contains all the $G_{{\mbox{\tiny\itshape SLOCC}}}$ invariants.
Furthermore, $G_{1,2,3}$ is isomorphic to ${\mathrm{SL(2)}}^{\times 3}$ and therefore it is possible to use Cayley's Omega Process (see e.g. \cite{olver}) to find the $G_{1,2,3}$ invariants. The Cayley Omega Process is an iterative method to construct the invariants from a representation of the states of the system by a set of multilinear forms. This procedure is described in the Appendix. 

While the Omega Process can be used to construct invariants it does not allow us to know if a given set of invariants generate the full algebra of SLOCC invariants. There is thus no guarantee that the seven constructed invariants are sufficient to generate all invariants. A simple dimension counting argument comparing the dimension of the Hilbert space with the dimension of the group $G_{{\mbox{\tiny\itshape SLOCC}}}$ gives that there must be at least 20-15=5 complex parameters describing the set of SLOCC interconvertibility classes. From this argument it is thus not ruled out that 7 invariants are sufficient.

Using the Omega Process a set of seven linearly independent invariants of degrees 4, 4, 8, 8, 8, 12, and 12 were found. The search did not yield any generators of degree 16. Higher degrees were not considered due to the rapidly increasing complexity of the search when the degrees increase. Thus, no conclusions can be made about the existence of higher degree invariants.

Since several independent invariants exist for each degree there is no unique choice of generators as homogeneous polynomials. The set was therefore chosen to give the individual invariants a clear physical meaning.
The two invariants of degree 4 were chosen as

\begin{align}
I^{(1)}&=m_{{\uparrow}0\diamondsuit}m_{0\diamondsuit{\uparrow}}(m_{\diamondsuit{\uparrow}0}m_{\downarrow\downarrow\downarrow}
-m_{\diamondsuit{\downarrow}0}m_{\downarrow\uparrow\downarrow})\nonumber\\
&+m_{{\uparrow}0\diamondsuit}m_{0\diamondsuit{\downarrow}}(m_{\diamondsuit{\downarrow}0}m_{\downarrow\uparrow\uparrow}
-m_{\diamondsuit{\uparrow}0}m_{\downarrow\downarrow\uparrow})\nonumber\\
&+m_{\downarrow 0\diamondsuit}m_{0\diamondsuit{\uparrow}}(
m_{\diamondsuit{\downarrow}0}m_{\uparrow\uparrow\downarrow}-m_{\diamondsuit{\uparrow}0}m_{\uparrow\downarrow\downarrow})\nonumber\\
&+m_{{\downarrow}0\diamondsuit}m_{0\diamondsuit{\downarrow}}(m_{\diamondsuit{\uparrow}0}m_{\uparrow\downarrow\uparrow}
-m_{\diamondsuit{\downarrow}0}m_{\uparrow\uparrow\uparrow}),
\end{align}
and
\begin{align}
I^{(2)}&=m_{{\uparrow}\diamondsuit 0}m_{\diamondsuit 0{\uparrow}}(m_{0{\uparrow}\diamondsuit}m_{\downarrow\downarrow\downarrow}
-m_{0{\downarrow}\diamondsuit}m_{\downarrow\uparrow\downarrow})\nonumber\\
&+m_{{\uparrow}\diamondsuit 0}m_{\diamondsuit 0{\downarrow}}(m_{0{\downarrow}\diamondsuit}m_{\downarrow\uparrow\uparrow}
-m_{0{\uparrow}\diamondsuit}m_{\downarrow\downarrow\uparrow})\nonumber\\
&+m_{\downarrow \diamondsuit 0}m_{\diamondsuit 0{\uparrow}}(
m_{0{\downarrow\diamondsuit }}m_{\uparrow\uparrow\downarrow}-m_{ 0{\uparrow}\diamondsuit}m_{\uparrow\downarrow\downarrow})\nonumber\\
&+m_{{\downarrow}\diamondsuit 0}m_{\diamondsuit 0{\downarrow}}(m_{0{\uparrow}\diamondsuit}m_{\uparrow\downarrow\uparrow}
-m_{0{\downarrow}\diamondsuit}m_{\uparrow\uparrow\uparrow}),
\end{align}
where $m_{ijk}$ is the amplitude of $|ijk\rangle$ in the state vector.
A state for which $I^{(1)}$ or $I^{(2)}$ is nonzero is Bell-nonlocal for every bi-partition, but does not necessarily show non-locality for a tripartition. This is thus the minimal required pattern of Bell-nonlocality necessary for the existence of SLOCC invariant entanglement as described in Theorem \ref{t1}. Note that $I^{(1)}$ or $I^{(2)}$ are related by a permutation of any pair of the spatial modes.

The three invariants of degree 8 were chosen as
\begin{widetext}

\begin{align}
I_{BC}=&(m_{0\diamondsuit\uparrow}m_{\diamondsuit 0\downarrow}-m_{0\diamondsuit\downarrow}m_{\diamondsuit 0\uparrow})(m_{0\uparrow\diamondsuit}m_{\diamondsuit\downarrow 0}-m_{0\downarrow\diamondsuit}m_{\diamondsuit\uparrow 0})\nonumber\\
&\times\Big[2m_{\uparrow 0\diamondsuit}m_{\uparrow\diamondsuit 0}(m_{\downarrow\downarrow\uparrow}m_{\downarrow\uparrow\downarrow}-m_{\downarrow\downarrow\downarrow}m_{\downarrow\uparrow\uparrow})+2m_{\downarrow 0\diamondsuit}m_{\downarrow\diamondsuit 0}(m_{\uparrow\downarrow\uparrow}m_{\uparrow\uparrow\downarrow}-m_{\uparrow\downarrow\downarrow}m_{\uparrow\uparrow\uparrow})\nonumber\\
&\phantom{\times[]}+(m_{\uparrow 0\diamondsuit}m_{\downarrow\diamondsuit 0}+m_{\downarrow 0\diamondsuit}m_{\uparrow\diamondsuit 0})(m_{\downarrow\downarrow\downarrow}m_{\uparrow\uparrow\uparrow}-m_{\downarrow\uparrow\downarrow}m_{\uparrow\downarrow\uparrow}-m_{\downarrow\downarrow\uparrow}m_{\uparrow\uparrow\downarrow}
+m_{\downarrow\uparrow\uparrow}m_{\uparrow\downarrow\downarrow})\Big],
\end{align}

\begin{align}
I_{AC}=&(m_{0\diamondsuit\uparrow}m_{\diamondsuit 0\downarrow}-m_{0\diamondsuit\downarrow}m_{\diamondsuit 0\uparrow})(m_{\uparrow 0\diamondsuit}m_{\downarrow\diamondsuit 0}-m_{\downarrow 0\diamondsuit}m_{\uparrow\diamondsuit 0})\nonumber\\
&\times\Big[2m_{0\uparrow\diamondsuit}m_{\diamondsuit\uparrow 0}(m_{\downarrow\downarrow\uparrow}m_{\uparrow\downarrow\downarrow}-m_{\downarrow\downarrow\downarrow}m_{\uparrow\downarrow\uparrow})+2m_{0\downarrow \diamondsuit}m_{\diamondsuit\downarrow 0}(m_{\downarrow\uparrow\uparrow}m_{\uparrow\uparrow\downarrow}-m_{\downarrow\uparrow\downarrow}m_{\uparrow\uparrow\uparrow})\nonumber\\
&\phantom{\times[]}+(m_{0\uparrow\diamondsuit}m_{\diamondsuit\downarrow 0}+m_{0\downarrow \diamondsuit}m_{\diamondsuit\uparrow 0})(m_{\downarrow\downarrow\downarrow}m_{\uparrow\uparrow\uparrow}-m_{\downarrow\uparrow\uparrow}m_{\uparrow\downarrow\downarrow}+
m_{\downarrow\uparrow\downarrow}m_{\uparrow\downarrow\uparrow}-m_{\downarrow\downarrow\uparrow}m_{\uparrow\uparrow\downarrow}
)\Big],
\end{align}
and

\begin{align}
I_{AB}=&(m_{0\uparrow\diamondsuit}m_{\diamondsuit\downarrow 0}-m_{0\downarrow\diamondsuit}m_{\diamondsuit\uparrow 0})(m_{\uparrow 0\diamondsuit}m_{\downarrow\diamondsuit 0}-m_{\downarrow 0\diamondsuit}m_{\uparrow\diamondsuit 0})\nonumber\\
&\times\Big[2m_{0\diamondsuit\uparrow}m_{\diamondsuit 0\uparrow}(m_{\downarrow\uparrow\downarrow}m_{\uparrow\downarrow\downarrow}-m_{\downarrow\downarrow\downarrow}m_{\uparrow\uparrow\downarrow})
+2m_{0\diamondsuit\downarrow }m_{\diamondsuit 0\downarrow}(m_{\downarrow\uparrow\uparrow}m_{\uparrow\downarrow\uparrow}-m_{\downarrow\downarrow\uparrow}m_{\uparrow\uparrow\uparrow})\nonumber\\
&\phantom{\times[]}+(m_{0\diamondsuit\uparrow}m_{\diamondsuit 0\downarrow}+m_{0 \diamondsuit\downarrow}m_{\diamondsuit 0\uparrow})(m_{\downarrow\downarrow\downarrow}m_{\uparrow\uparrow\uparrow}-m_{\downarrow\uparrow\uparrow}m_{\uparrow\downarrow\downarrow}-
m_{\downarrow\uparrow\downarrow}m_{\uparrow\downarrow\uparrow}+
m_{\downarrow\downarrow\uparrow}m_{\uparrow\uparrow\downarrow}
)\Big].
\end{align}
\end{widetext}

Any state for which $I_{BC}$ is nonzero displays nonlocal correlations between B and C in addition to the nonlocality over all bipartitons. If the state is projected onto the subspace where each party has a fermion, it is a separable state for A|BC but not for any other bipartition. Thus, if $I_{BC}$ is nonzero there exist local measurements with nonlocal correlations for parties B and C. 

The invariants $I_{AC}$ and $I_{AB}$ have the same meaning for the two other pairs of parties.
Furthermore, the three invariants $I_{AC},I_{BC}$, and $I_{AB}$ are independent of $I^{(1)}$ and $I^{(2)}$ in the sense that the set of invariants $\{I_{AB},I_{BC},I_{AC},I_1^2,I_2^2,I_1I_2\}$ does not have any element which is a linear combination of the other elements.

The two degree twelve invariants were chosen as

\begin{widetext}
\begin{align}
I_{ABC}^{(1)}=&(m_{0\diamondsuit\uparrow}m_{\diamondsuit 0\downarrow}-m_{0\diamondsuit\downarrow}m_{\diamondsuit 0\uparrow})(m_{0\uparrow\diamondsuit}m_{\diamondsuit\downarrow 0}-m_{0\downarrow\diamondsuit}m_{\diamondsuit\uparrow 0})
(m_{\uparrow 0\diamondsuit}m_{\downarrow\diamondsuit 0}-m_{\downarrow 0\diamondsuit}m_{\uparrow\diamondsuit 0})\nonumber\\
&\times\Big[2m_{\uparrow 0\diamondsuit}m_{0\diamondsuit\uparrow}(m_{\diamondsuit\downarrow 0}m_{\downarrow\uparrow\uparrow}-m_{\diamondsuit\uparrow 0}m_{\downarrow\downarrow\uparrow}) (m_{\downarrow\uparrow\downarrow}m_{\uparrow\downarrow\downarrow} - m_{\downarrow\downarrow\downarrow}m_{\uparrow\uparrow\downarrow})\nonumber\\
&\phantom{\times[]}+2m_{\uparrow 0\diamondsuit}m_{0\diamondsuit\downarrow}(m_{\diamondsuit\downarrow 0} m_{\downarrow\uparrow\downarrow}-m_{\diamondsuit\uparrow 0}m_{\downarrow\downarrow\downarrow}) (m_{\downarrow\uparrow\uparrow}m_{\uparrow\downarrow\uparrow} - m_{\downarrow\downarrow\uparrow}m_{\uparrow\uparrow\uparrow})\nonumber\\
&\phantom{\times[]}+2m_{\downarrow 0\diamondsuit}m_{0\diamondsuit\uparrow}(m_{\diamondsuit\downarrow 0} m_{\uparrow\uparrow\uparrow}-m_{\diamondsuit\uparrow 0}m_{\uparrow\downarrow\uparrow})
            (m_{\downarrow\downarrow\downarrow} m_{\uparrow\uparrow\downarrow}-m_{\downarrow\uparrow\downarrow}m_{\uparrow\downarrow\downarrow})\nonumber\\
&\phantom{\times[]}+2m_{\downarrow 0\diamondsuit}m_{0\diamondsuit\downarrow}(m_{\diamondsuit\downarrow 0}m_{\uparrow\uparrow\downarrow}-m_{\diamondsuit\uparrow{0}}m_{\uparrow\downarrow\downarrow}
)(m_{\downarrow\downarrow\uparrow}m_{\uparrow\uparrow\uparrow}-m_{\downarrow\uparrow\uparrow}m_{\uparrow\downarrow\uparrow})\nonumber\\
&\phantom{\times[]}+m_{\uparrow 0\diamondsuit}m_{0\diamondsuit\downarrow}(m_{\diamondsuit\downarrow 0}m_{\downarrow\uparrow\uparrow}-m_{\diamondsuit\uparrow 0}m_{\downarrow\downarrow\uparrow})(m_{\downarrow\downarrow\uparrow}m_{\uparrow\uparrow\downarrow} +m_{\downarrow\downarrow\downarrow}m_{\uparrow\uparrow\uparrow}-m_{\downarrow\uparrow\uparrow}m_{\uparrow\downarrow\downarrow} - m_{\downarrow\uparrow\downarrow}m_{\uparrow\downarrow\uparrow})\nonumber\\
&\phantom{\times[]}+m_{\uparrow 0\diamondsuit}m_{0\diamondsuit\uparrow}(m_{\diamondsuit\downarrow 0} m_{\downarrow\uparrow\downarrow}-m_{\diamondsuit\uparrow 0}m_{\downarrow\downarrow\downarrow})( m_{\downarrow\downarrow\uparrow}m_{\uparrow\uparrow\downarrow}+ 
           m_{\downarrow\downarrow\downarrow}m_{\uparrow\uparrow\uparrow}-m_{\downarrow\uparrow\uparrow}m_{\uparrow\downarrow\downarrow}- m_{\downarrow\uparrow\downarrow}m_{\uparrow\downarrow\uparrow})
           \nonumber\\
&\phantom{\times[]}+m_{\downarrow 0\diamondsuit}m_{0\diamondsuit\downarrow}(m_{\diamondsuit\downarrow 0} m_{\uparrow\uparrow\uparrow}-m_{\diamondsuit\uparrow 0}m_{\uparrow\downarrow\uparrow})(m_{\downarrow\uparrow\uparrow} m_{\uparrow\downarrow\downarrow} + m_{\downarrow\uparrow\downarrow}m_{\uparrow\downarrow\uparrow} - m_{\downarrow\downarrow\uparrow}m_{\uparrow\uparrow\downarrow} - 
m_{\downarrow\downarrow\downarrow}m_{\uparrow\uparrow\uparrow})\nonumber\\ 
&\phantom{\times[]}+m_{\downarrow 0\diamondsuit}m_{0\diamondsuit\uparrow}(m_{\diamondsuit\downarrow 0}m_{\uparrow\uparrow\downarrow}-m_{\diamondsuit\uparrow{0}}m_{\uparrow\downarrow\downarrow}
)(m_{\downarrow\uparrow\uparrow}m_{\uparrow\downarrow\downarrow}+ m_{\downarrow\uparrow\downarrow}m_{\uparrow\downarrow\uparrow}- m_{\downarrow\downarrow\uparrow}m_{\uparrow\uparrow\downarrow}- 
           m_{\downarrow\downarrow\downarrow}m_{\uparrow\uparrow\uparrow})\Big],
\end{align}
and
\begin{align}
I_{ABC}^{(2)}=&(m_{0\diamondsuit\uparrow}m_{\diamondsuit 0\downarrow}-m_{0\diamondsuit\downarrow}m_{\diamondsuit 0\uparrow})(m_{0\uparrow\diamondsuit}m_{\diamondsuit\downarrow 0}-m_{0\downarrow\diamondsuit}m_{\diamondsuit\uparrow 0})
(m_{\uparrow 0\diamondsuit}m_{\downarrow\diamondsuit 0}-m_{\downarrow 0\diamondsuit}m_{\uparrow\diamondsuit 0})\nonumber\\
&\times\Big[2m_{\uparrow\diamondsuit 0}m_{\diamondsuit 0\uparrow}(m_{0\downarrow\diamondsuit}m_{\downarrow\uparrow\uparrow}-m_{0\uparrow \diamondsuit}m_{\downarrow\downarrow\uparrow}) (m_{\downarrow\uparrow\downarrow}m_{\uparrow\downarrow\downarrow} - m_{\downarrow\downarrow\downarrow}m_{\uparrow\uparrow\downarrow})\nonumber\\
&\phantom{\times[]}+2m_{\uparrow\diamondsuit 0}m_{\diamondsuit 0\downarrow}(m_{0\downarrow \diamondsuit} m_{\downarrow\uparrow\downarrow}-m_{0\uparrow\diamondsuit}m_{\downarrow\downarrow\downarrow}) (m_{\downarrow\uparrow\uparrow}m_{\uparrow\downarrow\uparrow} - m_{\downarrow\downarrow\uparrow}m_{\uparrow\uparrow\uparrow})\nonumber\\
&\phantom{\times[]}+2m_{\downarrow\diamondsuit 0}m_{\diamondsuit 0\uparrow}(m_{0\downarrow\diamondsuit} m_{\uparrow\uparrow\uparrow}-m_{0\uparrow \diamondsuit}m_{\uparrow\downarrow\uparrow})
            (m_{\downarrow\downarrow\downarrow} m_{\uparrow\uparrow\downarrow}-m_{\downarrow\uparrow\downarrow}m_{\uparrow\downarrow\downarrow})\nonumber\\
&\phantom{\times[]}+2m_{\downarrow\diamondsuit 0}m_{\diamondsuit 0\downarrow}(m_{0\downarrow\diamondsuit}m_{\uparrow\uparrow\downarrow}-m_{0\uparrow\diamondsuit}m_{\uparrow\downarrow\downarrow}
)(m_{\downarrow\downarrow\uparrow}m_{\uparrow\uparrow\uparrow}-m_{\downarrow\uparrow\uparrow}m_{\uparrow\downarrow\uparrow})\nonumber\\
&\phantom{\times[]}+m_{\uparrow\diamondsuit 0}m_{\diamondsuit 0\downarrow}(m_{0\downarrow \diamondsuit}m_{\downarrow\uparrow\uparrow}-m_{0\uparrow \diamondsuit}m_{\downarrow\downarrow\uparrow})(m_{\downarrow\downarrow\uparrow}m_{\uparrow\uparrow\downarrow} +m_{\downarrow\downarrow\downarrow}m_{\uparrow\uparrow\uparrow}-m_{\downarrow\uparrow\uparrow}m_{\uparrow\downarrow\downarrow} - m_{\downarrow\uparrow\downarrow}m_{\uparrow\downarrow\uparrow})\nonumber\\
&\phantom{\times[]}+m_{\uparrow\diamondsuit 0}m_{\diamondsuit 0\uparrow}(m_{0\downarrow \diamondsuit} m_{\downarrow\uparrow\downarrow}-m_{0\uparrow \diamondsuit}m_{\downarrow\downarrow\downarrow})( m_{\downarrow\downarrow\uparrow}m_{\uparrow\uparrow\downarrow}+ 
           m_{\downarrow\downarrow\downarrow}m_{\uparrow\uparrow\uparrow}-m_{\downarrow\uparrow\uparrow}m_{\uparrow\downarrow\downarrow}- m_{\downarrow\uparrow\downarrow}m_{\uparrow\downarrow\uparrow})
           \nonumber\\
&\phantom{\times[]}+m_{\downarrow\diamondsuit 0}m_{\diamondsuit 0\downarrow}(m_{0\downarrow\diamondsuit} m_{\uparrow\uparrow\uparrow}-m_{0\uparrow \diamondsuit}m_{\uparrow\downarrow\uparrow})(m_{\downarrow\uparrow\uparrow} m_{\uparrow\downarrow\downarrow} + m_{\downarrow\uparrow\downarrow}m_{\uparrow\downarrow\uparrow} - m_{\downarrow\downarrow\uparrow}m_{\uparrow\uparrow\downarrow} - 
m_{\downarrow\downarrow\downarrow}m_{\uparrow\uparrow\uparrow})\nonumber\\ 
&\phantom{\times[]}+m_{\downarrow\diamondsuit 0}m_{\diamondsuit 0\uparrow}(m_{0\downarrow\diamondsuit}m_{\uparrow\uparrow\downarrow}-m_{0\uparrow\diamondsuit}m_{\uparrow\downarrow\downarrow}
)(m_{\downarrow\uparrow\uparrow}m_{\uparrow\downarrow\downarrow}+ m_{\downarrow\uparrow\downarrow}m_{\uparrow\downarrow\uparrow}- m_{\downarrow\downarrow\uparrow}m_{\uparrow\uparrow\downarrow}- 
           m_{\downarrow\downarrow\downarrow}m_{\uparrow\uparrow\uparrow})\Big].
\end{align}
\end{widetext}
A state for which $I_{ABC}^{(1)}$ or $I_{ABC}^{(2)}$ is nonzero features Bell-nonlocality between any pair of parties. A projection of such a state onto the subspace where there is one particle with each party is in either the GHZ SLOCC interconvertibility class or in the W SLOCC interconvertibility class of three distinguishable spin-$\frac{1}{2}$ particles. These two SLOCC classes are the only two classes for such a system where the states in the class are not separable for any partition \cite{dur}. 
Moreover, the two invariants $I_{ABC}^{(1)}$ and $I_{ABC}^{(2)}$ are independent of the lower degree invariants in the sense that the set of invariants $\{I_{ABC}^{(1)},I_{ABC}^{(2)}, I_2I_{AB},I_2I_{BC},I_2I_{AC},I_1I_{AB},I_1I_{BC},I_1I_{AC},$ $I_1^3,I_2^3,I_1^2I_2,I_2^2I_1\}$ does not have any element which is a linear combination of the other elements.

Note that for states that are symmetric under permutations of the spatial modes, the invariants $I_{AB},I_{BC},I_{AC},I_{ABC}^{(1)},I_{ABC}^{(2)}$ take the value zero.

Measures corresponding to the individual invariants can be constructed as $|I^{(1)}|^{1/2},|I^{(2)}|^{1/2},|I_{AB}|^{1/4},$ $|I_{BC}|^{1/4},|I_{AC}|^{1/4},|I_{ABC}^{(1)}|^{1/6}$, and $|I_{ABC}^{(2)}|^{1/6}$, respectively. But any homogeneous polynomial formed by combining the invariants can be used to construct a measure. For example, the invariants $I^{(1)}+I^{(2)}$ and $I_{AB}+I_{AC}+I_{BC}$ can be used to construct measures $|I^{(1)}+I^{(2)}|^{1/2}$ and $|I_{AB}+I_{AC}+I_{BC}|^{1/4}$ that are invariant under permutations of the spatial modes.

\subsection{Examples of maximally entangled states}
To better understand the different types of entanglement parametrised by the invariants, we can consider a few examples of maximally entangled states for which only one single SLOCC invariant takes a nonzero value.
A maximally entangled state for which $I^{(2)}\neq 0$ while $I^{(1)}=I_{AB}=I_{BC}=I_{AC}=I_{ABC}^{(1)}=I_{ABC}^{(2)}=0$ is

\begin{eqnarray}
\frac{1}{2}\big(|\!\uparrow\!\diamondsuit 0\rangle+|0\!\uparrow\!\diamondsuit\rangle+|\diamondsuit 0{\uparrow}\rangle+|\!\downarrow\downarrow\downarrow\rangle\big).
\end{eqnarray}
Thus, this state only has entanglement of the type distinguished by $I^{(2)}$.
For the tri-partitioning A|B|C no measurements can distinguish this state from a separable state. However, entanglement is detectable for any bipartition of the system, i.e., for A|BC, B|AC, and C|AB. Moreover, the measure $|I^{(2)}|^{1/2}$ takes its maximal value $1/4$ for this state. A state for which only $I^{(1)}$ is nonzero can be found by permuting any pair of the parties.

In a similar way we can construct a state for which the only nonzero invariant is $I_{AB}$,
 
\begin{align}
\frac{1}{\sqrt{8}}\big(&|\!\uparrow\!\diamondsuit 0\rangle+|0{\uparrow}\diamondsuit\rangle+|\diamondsuit 0{\uparrow}\rangle
+|0\diamondsuit \!\uparrow\rangle\nonumber\\&+|\!\downarrow\!0\diamondsuit\rangle+|\diamondsuit{\downarrow}0\rangle
+|\!\downarrow\uparrow\downarrow\rangle+|\!\uparrow\downarrow\downarrow\rangle\big).
\end{align}
This state  contains an amplitude of a Bell-state between A and B, $1/\sqrt{2}(|\!\downarrow\uparrow\rangle+|\!\uparrow\downarrow\rangle)\otimes|\!\downarrow\rangle$. Therefore apart from Bell-nonlocality across every bipartition, Bell-nonlocality is also exhibited between A and B but not between A and C or B and C. States that give nonzero values only to $I_{AC}$ or $I_{BC}$ can be constructed by permuting parties.

A state for which only $I_{ABC}^{(1)}$ is nonzero is
\begin{align}
\frac{1}{\sqrt{12}}\big[&\sqrt{2}(|\!\uparrow\! 0 \diamondsuit\rangle+|0\diamondsuit{\uparrow}\rangle+|\diamondsuit {\uparrow}0\rangle)\nonumber\\
&+|\diamondsuit 0\!\downarrow\rangle+|0\!\downarrow\!\diamondsuit\rangle+|\!\downarrow\!\diamondsuit 0\rangle\nonumber\\
&+|\!\downarrow\uparrow\downarrow\rangle+|\!\uparrow\downarrow\downarrow\rangle+|\!\downarrow\downarrow\uparrow\rangle\big].
\end{align}
This state contains an amplitude of the $W$ state $1/\sqrt{3}(|\!\downarrow\uparrow\downarrow\rangle+|\!\uparrow\downarrow\downarrow\rangle+|\!\downarrow\downarrow\uparrow\rangle)$ and therefore features nonlocality for every pair of parties. A similar state for which only $I_{ABC}^{(2)}$ is non-zero can be constructed by permuting any two parties.

\subsection{Arbitrary number of fermions and arbitrary spin}

In Sect. \ref{ent} the relation between SLOCC invariants and Bell-nonlocality, was given in Theorem \ref{t1}. Furthermore, in Theorem \ref{lut} the relation between the number of fermions, $m$, and the number of spatial modes, $n$, that is needed for the existence of entanglement measures and maximally entangled states was given. 
These relations can be generalized to higher spin and to larger numbers of fermions. 

First we consider the generalization of Theorem \ref{lut}.

\begin{theorem}
\label{thee2}
In a system of $n$ spatial modes and $m$ delocalized spin-$\frac{p}{2}$ fermions, where $m\geq{p+1}$, entanglement measures and maximally entangled states exist if and only if $m/n=(p+1)/2$.
\end{theorem}
\begin{proof}
To show necessity of the relation $m/n=(p+1)/2$ we construct the generalization of the proof of Theorem \ref{lut}.
Consider $m$ delocalized spin-$\frac{p}{2}$ fermions in a system of $n$ spatial modes.

Each fermion has $p+1$ internal degrees of freedom. Therefore the subspace of the local Hilbert space corresponding to one fermion in the spatial mode has dimension $p+1$. The subspace corresponding to $k$ fermions in a spatial mode has a dimension given by the binomal coefficient $\binom{p+1}{k}$ due to the Pauli exclusion principle. The full Hilbert space of a spatial mode therefore has dimension $\sum_{k=0}^{p+1}\binom{p+1}{k}=2^{p+1}$.  

Consider the diagonal $2^{p+1}\times 2^{p+1}$ matrix $g$ acting on the local Hilbert space of some spatial mode

\begin{eqnarray}g=\left( \begin{array}{cccccc}
\! {\bf B}_{0} & \! 0 & 0 & \cdots & 0 & 0\\
\! 0 & {\bf B}_{{1}} & 0 & \cdots & 0 & 0\\
\! 0 & 0 &  \ddots &  & \vdots & \vdots \\
\! \vdots & \vdots & & \ddots &  0 & 0\\
\! 0 &  0   & \cdots & 0 & {\bf B}_{p} & 0\\
\! 0 & 0  & \cdots & 0 & 0 & {\bf B}_{p+1}\end{array} \!\right),\nonumber
\end{eqnarray}
where ${\bf B}_{k}$ is a $\binom{p+1}{k}\times\binom{p+1}{k}$, block corresponding to $k$ fermions in the spatial mode, where all off-diagonal elements are zero and all diagonal elements are equal to $e^{(1-\frac{2k}{p+1})\phi}$ for $\phi\in\mathbb{C}$.

The operation $g^{\times{n}}$ belongs to the $G_{{\mbox{\tiny\itshape SLOCC}}}$ group for spin-$\frac{p}{2}$ fermions since $det(g)=1$.
The action of $g^{\times{n}}$ on any basis vector of the Hilbert space depends only on the number $m$ of fermions and the number $n$ of spatial modes. More precisely, this action is a scaling factor $e^{[n-2m/(p+1)]\phi}$ of the state vector.

It follows that for $Re(\phi)\neq{0}$ it is impossible to construct polynomials in the state vector coefficients that are invariant under this action unless $m/n=(p+1)/2$. 

For the sufficiency of the relation $m/n=(p+1)/2$ we make an explicit construction of maximally entangled states for arbitrary $n,m,$ and $p$ satisfying the relation.
For a given $p$, consider the sequence $s(p+1)\equiv 0,1,2,\dots,2^{p+1}-3,2^{p+1}-2,2^{p+1}-1$, where $0$ represents the absence of a fermion and the other entries represent the different basis vectors of the local Hilbert space where one or more fermions occupy the mode. Let $\sigma$ be the cyclic permutation defined by
$\sigma(0,1,2,\dots,2^{p+1}-3,2^{p+1}-2,2^{p+1}-1)=2^{p+1}-1,0,1,2,\dots,2^{p+1}-3,2^{p+1}-2$. Furthermore, let $\sigma^{k}$ denote the composition of $k$ such cyclic permutations.

Each sequence of this kind corresponds to a Hilbert space vector which is a product of local states as given by the sequence.
Given this we can define a state $|\Psi_{\sigma}^{p+1}\rangle$ by
\begin{eqnarray}
|\Psi_{\sigma}^{p+1}\rangle=\frac{1}{\sqrt{2^{p+1}}}\sum_{i=1}^{2^{p+1}}|\sigma^{i}(0,1,\dots,2^{p+1}-2,2^{p+1}-1)\rangle,\nonumber
\end{eqnarray}
which is a maximally entangled state of $2^{p}(p+1)$ fermions with spin $\frac{p}{2}$ delocalized over $2^{p+1}$ spatial modes. To verify this we calculate the single-party reduced density matrices and see that they are all maximally mixed.
For any spatial mode the state $|\Psi_{\sigma}^{p+1}\rangle$ can be written as 

\begin{eqnarray}\label{label}
|\Psi_{\sigma}^{p+1}\rangle=\frac{1}{\sqrt{2^{p+1}}}\sum_{i=0}^{2^{p+1}-1}|i\rangle|\theta_i\rangle,
\end{eqnarray}
where $\langle\theta_i|\theta_j\rangle=0$ if $i\neq j$. Thus, the single spatial mode reduced density matrices are diagonal. Moreover, since each term in Eq. \ref{label} comes with the same coefficient $1/\sqrt{2^{p+1}}$, each of the reduced density matrices are proportional to the identity.

To construct states of $r2^{p}(p+1)$ fermions with spin $\frac{p}{2}$ delocalized over $r2^{p+1}$ spatial modes we form a sequence $s(r,p+1)\equiv s(p+1)s(p+1)\dots s(p+1)$ by concatenating $r$ copies of $s(p+1)$, i.e., $s(r,p+1)\equiv 0,1,\dots,2^{p+1}-2,2^{p+1}-1,0,1,\dots,2^{p+1}-2,2^{p+1}-1,0,1,\dots$. From this sequence we then define the state

\begin{eqnarray}
|\Psi_{\sigma}^{r,p+1}\rangle=\frac{1}{\sqrt{2^{p+1}}}\sum_{i=1}^{2^{p+1}}|\sigma^{i}[s(r,p+1)]\rangle.\nonumber
\end{eqnarray}

In this way we can construct maximally entangled states for any $n,m,$ and $p$ that satisfies $m/n=(p+1)/2$.
\end{proof}
The simplifying assumption $m\geq{p+1}$ is made to ensure that the local Hilbert space of each spatial mode is not constrained by the total particle number in the system.
From Theorem \ref{thee2} we can conclude that, given $m\geq{p+1}$, the particle species determines the required ratio of particles to the number of spatial modes for the existence of maximal entanglement.
Since for fermions $p$ is an odd number, the filling fraction $m/n=(p+1)/2$ of the spatial modes is always an integer. 

Next we consider the generalization of Theorem \ref{t1}.

\begin{theorem}\label{t1g}
Let $|\psi\rangle$ be a state of $m$ delocalized spin-$\frac{p}{2}$ fermions in $n$ modes where $m/n=(p+1)/2$. If $|\psi\rangle$ is Bell-local over a bipartition of the system into a single mode and $n-1$ modes, there exist no entanglement monotones that take a nonzero value.
This implies that $|\psi\rangle$ cannot be converted to a maximally entangled state by SLOCC. Moreover, local special unitary operations can induce an arbitrary global phase shift of the state vector $|\psi\rangle$. Thus, if an entanglement monotone takes a non-zero value the state is Bell-nonlocal over every bi-partition of this type.
\end{theorem}
\begin{proof}

For a partitioning into a single spatial mode and $n-1$ modes, a general state $|\psi\rangle$ can be written as
\begin{eqnarray}|\psi\rangle=\sum_{k=1}^{p+1}\sum_{i_{k}=1}^{\binom{p+1}{k}}r_{i_{k}}|i_{k}\rangle|\theta_{i_{k}}\rangle+r_0|0\rangle|\phi\rangle,\nonumber\end{eqnarray} 
where the $|i_k\rangle$ is a basis of the subspace of the local Hilbert space of the single mode corresponding to $k$ fermions in the mode, and $|0\rangle$ represents no fermion in the mode. The $|\theta_{i_{k}}\rangle$ are states of $m-k$ fermions in $n-1$ modes and $|\phi\rangle$ is a state of $m$ fermions in $n-1$ modes. Consider the projection of the state onto a subspace where $k$ fermions are localized in the single-mode part of the partitioning and $m-k$ in the remaining $n-1$ modes. Every state in the subspace spanned by the $|\theta_{i_{k}}\rangle$ contains the same number of fermions. Therefore, within this subspace the operations by a party with access to the $n-1$ partition are not restricted by the super-selection rule. Thus, nonlocal correlations in this superselection sector can be observed if and only if the projection of the state to this sector is entangled. 

Therefore, the state is Bell-nonlocal if for some $k$ the state satisfies $|\theta_{i_{k}}\rangle\neq|\theta_{j_{k}}\rangle$ and $r_{i_{k}}r_{j_{k}}\neq 0$ for some ${i_{k}}$ and ${j_{k}}$.
If the state is Bell-local it thus follows that $r_{i_{k}}r_{j_{k}}=0$ or $|\theta_{i_{k}}\rangle=|\theta_{j_{k}}\rangle$ for each pair of ${i_{k}},{j_{k}}$ and each $k$. Such a state is, up to local unitary operations, of the form $|\psi'\rangle=\sum_{k=1}^{p+1}s_{1_k}|1_{k}\rangle|\varphi_k\rangle+s_2|0\rangle|\phi\rangle$, where ${1_k}$ is a given state of $k$ fermions in the mode, and is indistinguishable from a separable state $\sum_{k=1}^{p+1}|s_{1_{k}}|^2|{1_k}\rangle\langle {1_k}|\otimes|\varphi_k\rangle\langle\varphi_k|+|s_2|^2|0\rangle\langle 0|\otimes|\phi\rangle\langle\phi|$.

Next consider the diagonal $2^{p+1}\times 2^{p+1}$ determinant one matrix

\begin{eqnarray}g(\alpha)=\left( \begin{array}{ccccccc}
\! {\bf B}_{1} & \! 0 & 0 & \cdots & 0 & 0 & 0\\
\! 0 & {\bf B}_{{2}} & 0 & \cdots & 0 & 0 & 0\\
\! 0 & 0 &  \ddots &  & \vdots & \vdots & \vdots\\
\! \vdots & \vdots & & \ddots &  0 & 0 & 0\\
\! 0 &  0   & \cdots & 0 & {\bf B}_{p} & 0 & 0\\
\! 0 & 0  & \cdots & 0 & 0 & {e^{(1-2^{p})\alpha}} & 0\\
\! 0 & 0  & \cdots & 0 & 0 & 0 & {e^{(1-2^{p})\alpha}}\end{array} \!\right),\nonumber
\end{eqnarray}
where the last two columns corresponds to the action on $|0\rangle$ and $|1_{p+1}\rangle$, respectively, where $|1_{p+1}\rangle$ is the state of $p+1$ fermions in the mode,
and each ${\bf B}_{k}$ is a $\binom{p+1}{k}\times{\binom{p+1}{k}}$ block

\begin{eqnarray}{\bf B}_{k}=\left( \begin{array}{ccccc}
\!\! e^{(1-2^{p})\alpha} & \! 0 & 0 & \cdots & 0 \\
\!\! 0 & e^{(1+2^{p})\alpha} & 0 & \cdots & 0 \\
\!\! 0 & 0 &  e^{\alpha} &  & \vdots  \\
\!\! \vdots & \vdots & & \ddots &  0 \\
\!\! 0 &  0   & \cdots & 0 & e^{\alpha} \\
\end{array} \!\right),\nonumber
\end{eqnarray}
where $\alpha\in\mathbb{C}$. Here the first column corresponds to the action on the local state $|1_k\rangle$.

For a state of the form  $\sum_{k=1}^{p+1}s_{1_k}|1_{k}\rangle|\varphi_k\rangle+s_2|0\rangle|\phi\rangle$ the SLOCC operation $g(\alpha)\times{I}\dots\times{I}$ scales the state vector by a factor of $|e^{(1-2^{p})\alpha}|$ and changes the phase by $\arg[e^{(1-2^{p})\alpha}]$. By making $Re(\alpha)$ sufficiently large we can bring the norm arbitrarily close to zero. This implies that the set of states which are inter-convertible with $|\psi'\rangle$ does not contain any maximally entangled states \cite{kempfness,verstraete2003} and that no polynomial SLOCC invariant exists. This follows since no homogeneous polynomial is invariant under a scaling of the state vector.
Thus, if there exist an entanglement monotone which takes a non-zero value for a state this state cannot have the form $\sum_{k=1}^{p+1}s_{1_k}|1_{k}\rangle|\varphi_k\rangle+s_2|0\rangle|\phi\rangle$ for any bipartition into a single mode and $n-1$ modes. Thus, the state must be Bell-nonlocal over every such bipartition.

The form of the state $|\psi'\rangle$ also implies that when $Re(\alpha)=0$ the SLOCC operation is special unitary and under this operation $g(\alpha)\times{I}\dots\times{I}$ the state accumulates a phase-factor $e^{(1-2^{p})\alpha}$.
Thus, $e^{(2^{p}-1)\alpha}g(\alpha)\times{I}\dots\times{I}$ is an element of the stabilizer group.

\end{proof}

Theorem \ref{t1g} shows that just as in the case of spin-$\frac{1}{2}$ fermions the existence of SLOCC invariants and monotones implies that there exist bipartitions for which the system is Bell-nonlocal.

\section{Additional constraints}
\label{sec4}
Section \ref{sec3} described entanglement measures 
in the case of unconstrained dynamics where any spatial configuration of the fermions compatible with the Pauli exclusion principle was allowed. We can however consider cases where the dynamics is more constrained. Here we consider the cases of strongly repulsive and strongly attractive interaction between the fermions and the case where one or more of the fermions are in fixed spatial locations.

In these cases the more constrained dynamics leads to a lower dimensional local Hilbert space for one or more of the spatial modes, and the group of local operations is therefore also constrained. This changes the notion of maximal entanglement as well as the role of invariants under the local operations and the corresponding measures.

\subsection{Strong repulsive or attractive interaction}

We first consider the case of spin-$\frac{1}{2}$ fermions with either strong repulsive or strong attractive interaction.
The relation between SLOCC invariants and Bell-nonlocality described in Theorem \ref{t1} carries over to the case of strongly repulsive interaction between the fermions. However, in the case of strongly attractive interaction the same simple relation does not hold.

\subsubsection{Strongly repulsive interaction}
In the limit of strong repulsive interaction between the fermions, double occupancies of a spatial mode are effectively not possible and the local Hilbert space of any mode is spanned by only $|\!\uparrow\rangle,|\!\downarrow\rangle,$ and $|0\rangle$.

A strong repulsion between the fermions can be caused by for example Coulomb interaction or in the case of composite fermions the Fermi-repulsion between the constituent particles \cite{dyson}.
If the composite Fermions have an effective size that is large compared to the size of the spatial modes the Fermi repulsion can prevent double occupancies.

Entanglement in the case of strong repulsion was considered in Ref. \cite{johansson15} for the case of two particles. With this constraint the degree of any SLOCC invariant polynomial is a multiple of three. In this case the full set of invariants is generated by only two invariants of degree 3 and 6, respectively. These are

\begin{align}\label{i1}
 I_1 =&(m_{\uparrow\uparrow{0}}m_{\downarrow\downarrow{0}}-m_{\uparrow\downarrow{0}}m_{\downarrow\uparrow{0}})\nonumber\\
 \times &(m_{0\uparrow\uparrow}m_{0\downarrow\downarrow}-m_{0\uparrow\downarrow}m_{0\downarrow\uparrow})\nonumber\\
 \times &(m_{\uparrow{0}\uparrow}m_{\downarrow{0}\downarrow}-m_{\uparrow{0}\downarrow}m_{\downarrow{0}\uparrow}),
\end{align}
and
\begin{align}\label{i4}
I_{2}=&m_{\uparrow\uparrow{0}}(m_{{0}\downarrow\downarrow}m_{\downarrow{0}\uparrow}-m_{{0}\downarrow\uparrow}m_{\downarrow{0}\downarrow})\nonumber\\+&m_{\uparrow\downarrow{0}}(m_{0\uparrow\uparrow}m_{\downarrow{0}\downarrow}-m_{{0}\uparrow\downarrow}m_{\downarrow{0}\uparrow})\nonumber\\+&m_{\downarrow\uparrow{0}}(m_{{0}\downarrow\uparrow}m_{\uparrow{0}\downarrow}-m_{{0}\downarrow\downarrow}m_{\uparrow{0}\uparrow})\nonumber\\+&m_{\downarrow\downarrow{0}}(m_{{0}\uparrow\downarrow}m_{\uparrow{0}\uparrow}-m_{0\uparrow\uparrow}m_{\uparrow{0}\downarrow}).
\end{align}

Just as in the unconstrained case a nonzero value of any of these two invariants implies that the state is Bell-nonlocal for all bipartitions. This property remains true also for higher spin and larger numbers of particles, i.e., if a SLOCC invariant is nonzero there the state is Bell-nonlocal across all bipartitions into a single spatial mode and the rest of the modes, as shown in Ref. \cite{johansson15}.

\subsubsection{Strongly attractive interaction}
\label{att}
For completeness we consider the case of spin-$\frac{1}{2}$ fermions with strong attractive interaction as well. The entanglement properties of such a system depends on whether the number of fermions is odd or even.

The case of even fermion number serves as an example where a nonzero value of a SLOCC invariant does not imply that the state is Bell-nonlocal. 
If the particle interaction is strongly attractive and the number of spin-$\frac{1}{2}$ fermions is even, the local Hilbert space of a spatial mode is effectively reduced to the span of $|0\rangle$ and $|\diamondsuit\rangle$. Due to the particle number superselection rule the local action of the group of SLOCC is described by $2\times 2$ diagonal determinant one matrices. 

For two parties sharing two fermions the SLOCC invariants are generated by the monomial 

\begin{align}\label{i5}
I=m_{0\diamondsuit}m_{\diamondsuit{0}}.
\end{align}
In this simple case it is impossible to observe Bell-nonlocal correlations since the Hilbert space is spanned by the two vectors $|0\diamondsuit\rangle$ and $|\diamondsuit{0}\rangle$ which have different local particle numbers. Thus, in contrast to the case of strongly repulsive interaction a SLOCC invariant that takes a nonzero value does not imply that the state is Bell-nonlocal.

Bell-nonlocality can only be observed if a larger system is considered and each partition contains at least two spatial modes.
For a system of four fermions and four parties all invariants are generated by the three degree 2 invariants, $m_{0\diamondsuit\diamondsuit 0}m_{\diamondsuit{0}0\diamondsuit},m_{0\diamondsuit 0\diamondsuit }m_{\diamondsuit{0}\diamondsuit 0},$ and $m_{00\diamondsuit\diamondsuit}m_{\diamondsuit\diamondsuit{0}{0}}$. The states in a system of this type are Bell local for the quadru-partition A|B|C|D, as well as any tripartition, e.g., AB|C|D. For bipartitions with two parties in each partition Bell-nonlocal states can be constructed. Two parties in each partition are needed so that the Hilbert space of each part has a two-dimensional two-fermion subspace. Then the entanglement and Bell-nonlocality can be analysed analogously to the case of two localized spin-$\frac{1}{2}$ particles with the spatial position of a fermion pair replacing the spin degree of freedom. 
A given state may be Bell-nonlocal for more than one bipartition.
For example  $1/\sqrt{2}(|0\diamondsuit\diamondsuit{0}\rangle+|\diamondsuit{0}0\diamondsuit\rangle)$ is Bell-nonlocal for the partition AB|CD as well as the partition AC|BD.

Although a nonzero SLOCC invariant in general is no guarantee for Bell-nonlocality across a bipartition, special invariants with this property can be constructed. These are

\begin{align}\label{i6}
I_{AB|CD}=m_{0\diamondsuit\diamondsuit 0}m_{\diamondsuit{0}0\diamondsuit}-m_{0\diamondsuit0\diamondsuit}m_{\diamondsuit{0}\diamondsuit{0}},\nonumber\\
I_{AD|BC}=m_{00\diamondsuit\diamondsuit}m_{\diamondsuit\diamondsuit{0}{0}}-m_{0\diamondsuit0\diamondsuit}m_{\diamondsuit{0}\diamondsuit{0}},\nonumber\\
I_{AC|BD}=m_{0\diamondsuit\diamondsuit 0}m_{\diamondsuit{0}0\diamondsuit}-m_{00\diamondsuit\diamondsuit}m_{\diamondsuit\diamondsuit{0}{0}}.
\end{align}
Here $I_{AB|CD}$ is essentially the concurrence \cite{wootters} across AB|CD and a nonzero value guarantees the existence of Bell-nonlocal correlations. The invariants $I_{AD|BC}$ and $I_{AC|BD}$ have the analogous meaning for the two other bipartitions.

If the number of fermions is odd there will still remain one unpaired fermion. This implies that the local Hilbert space of any mode is spanned by $|\!\uparrow\rangle,|\!\downarrow\rangle,|0\rangle$, and $|\diamondsuit\rangle$, and that there exist states that are Bell-nonlocal across all bipartitions into one spatial mode and the rest, e.g. A|BCD. An example of a three mode state that is Bell-nonlocal across all bipartitions is 

\begin{eqnarray}
\frac{1}{\sqrt{6}}(|\!\uparrow\!\diamondsuit 0\rangle+|\!\downarrow\!0\diamondsuit \rangle+|0\!\uparrow\!\diamondsuit\rangle+|\diamondsuit\!\downarrow\!0\rangle+|\diamondsuit 0{\uparrow}\rangle+|0\diamondsuit{\downarrow}\rangle).\nonumber\\
\end{eqnarray}

However, maximally entangled states do not exist and therefore there are no SLOCC invariants. It can be shown that it is impossible to choose the state vector coefficients so that all reduced density matrices are maximally mixed, and that it is impossible to form polynomials that are invariant under the action of $G_{8,15}$.

\subsection{Partial localization}

Some physical systems can be modelled as a combination of localized and mobile particles. There exist materials where interaction between the localized and mobile particles is important for physical properties. See e.g. Ref. \cite{lyra} for a theoretical study and Refs. \cite{mukuda, ishiwata, ishiwata2,carter, hiroi} for examples of layered transition metal oxides with this kind of interaction.

In a system with such partial localization, the local Hilbert space of a spatial mode where a number of particles are localized is spanned by a smaller number of vectors than in the unconstrained case. Therefore, the local action of the group of SLOCC has a smaller number of generators. Just as in the unconstrained case, a nonzero value of a SLOCC invariant implies that the state is Bell-nonlocal over a bipartition into one spatial mode and the other spatial modes, as long as the dimension of the local Hilbert space of the single mode is greater than one. This can be seen by a small modification of the proof of Theorem \ref{t1g}. For clarity we give it as a separate corollary. 

\begin{corollary}\label{t1g2}
Let $|\psi\rangle$ be a state of $m$ spin-$\frac{p}{2}$ fermions in $n$ modes where some of the spatial modes contain a fixed number of fermions. Assume that no mode with a fixed number of fermions contain 0 or $p+1$ fermions. Then, if $|\psi\rangle$ is Bell-local over a bipartition of the system into a single mode and $n-1$ modes, there exist no entanglement monotones that take a nonzero value.
This implies that $|\psi\rangle$ cannot be converted to a maximally entangled state by SLOCC. Moreover, local special unitary operations can induce an arbitrary global phase shift of the state vector $|\psi\rangle$. Thus, if an entanglement monotone takes a non-zero value the state is Bell-nonlocal over every bi-partition of this type.
\end{corollary}
\begin{proof}

For a partitioning into a single mode and $n-1$ modes in the case where the number of particles in the single mode is not fixed, the proof is identical to that of Theorem \ref{t1g}. Consider therefore a mode that contains a fixed nonzero number $0<k<p+1$ of particles. For a partitioning into this single mode and the other $n-1$ modes, a general state $|\psi\rangle$ can be written as

\begin{eqnarray}|\psi\rangle=\sum_{i_{k}=1}^{\binom{p+1}{k}}r_{i_{k}}|i_{k}\rangle|\theta_{i_{k}}\rangle,\nonumber\end{eqnarray} 
where the $|i_k\rangle$ is a basis of the $k$-fermion local Hilbert space of the single mode. The $|\theta_{i_{k}}\rangle$ are states of $m-k$ fermions in $n-1$ modes and thus all states in the subspace spanned by the $|\theta_{i_{k}}\rangle$ contains the same number of particles. Therefore, the operations by a party with access to the $n-1$ partition are not restricted by the super-selection rule. 

Nonlocal correlations can be observed if and only if the state is entangled. 
More precisely, the state is Bell-nonlocal if it is true that $|\theta_{i_{k}}\rangle\neq|\theta_{j_{k}}\rangle$ for some ${i_{k}}$ and ${j_{k}}$ and $r_{i_{k}}r_{j_{k}}\neq 0$.
If the state is Bell-local it thus follows that $r_{i_{k}}r_{j_{k}}=0$ or $|\theta_{i_{k}}\rangle=|\theta_{j_{k}}\rangle$ for each pair of ${i_{k}},{j_{k}}$. Such a state is, up to local unitary operations, of the form $|1_{k}\rangle|\varphi\rangle$, where ${1_k}$ is a given state of $k$ particles in the mode, and is separable.

Next consider the diagonal $\binom{p+1}{k}\times{\binom{p+1}{k}}$ determinant one matrix

\begin{eqnarray}{g}_{\alpha}=\left( \begin{array}{ccccc}
\!\! e^{-\alpha} & \! 0 & 0 & \cdots & 0 \\
\!\! 0 & e^{\alpha} & 0 & \cdots & 0 \\
\!\! 0 & 0 &  1 &  & \vdots  \\
\!\! \vdots & \vdots & & \ddots &  0 \\
\!\! 0 &  0   & \cdots & 0 & 1 \\
\end{array} \!\right),\nonumber
\end{eqnarray}
where $\alpha\in\mathbb{C}$. Here the first column corresponds to the action on the local state $|1_k\rangle$.

For a state of the form  $|1_{k}\rangle|\varphi\rangle$ the SLOCC operation $g(\alpha)\times{I}\dots\times{I}$ scales the state vector by a factor of $|e^{-\alpha}|$ and changes the phase by $\arg(e^{-\alpha})$. By making $Re(\alpha)$ sufficiently large we can bring the norm of the state arbitrarily close to zero. This implies that the set of states which are inter-convertible with $|1_{k}\rangle|\varphi\rangle$ does not contain any maximally entangled states \cite{kempfness,verstraete2003} and that no polynomial SLOCC invariant takes a nonzero value. This follows since no homogeneous polynomial is invariant under a scaling of the state vector.
Thus, if there exist an entanglement monotone which takes a non-zero value for a state, this state cannot have the form $|1_{k}\rangle|\varphi\rangle$ for any bipartition into a single mode and $n-1$ modes. Thus, the state must be Bell-nonlocal over every such bipartition.

The form of the state $|1_{k}\rangle|\varphi\rangle$ also implies that when $Re(\alpha)=0$ the SLOCC operation is special unitary and under this operation $g(\alpha)\times{I}\dots\times{I}$ the state accumulates a phase-factor $e^{-\alpha}$.
Thus, $e^{\alpha}g(\alpha)\times{I}\dots\times{I}$ is an element of the stabilizer group.

\end{proof}

Note that in the case where $p+1$ spin-$\frac{p}{2}$ fermions are fixed in the same spatial mode, which Corollary \ref{t1g2} does not consider, the state of the system is always a product state of the form $|\diamondsuit_{p+1}\rangle|\phi\rangle$ where $|\diamondsuit_{p+1}\rangle$ is the state of $p+1$ fermions. A state of this type is thus trivially Bell-local for the partitioning into the given spatial mode and the other modes. In this case the only determinant one operation on the local Hilbert space that is spanned by the single vector $|\diamondsuit_{p+1}\rangle$ is the identity operation and thus no SLOCC operation can scale the state vector. Therefore, the existence of SLOCC invariants that take a nonzero value cannot be excluded by the argument used in the corollary.

\subsubsection{One localized and two delocalized spin-$\frac{1}{2}$ fermions and three parties}\label{partloc}

The simplest case that allows for partial localization is three spatial modes. We consider such a system containing three spin-$\frac{1}{2}$ fermions. 
If one particle is constrained to be localized in the spatial mode of A, the local Hilbert space of A is spanned by $|\!\uparrow\rangle,|\!\downarrow\rangle$.
Therefore local action of the group of SLOCC on the spatial mode of A is generated by only $\lambda_1,\lambda_2,$ and $\lambda_3$. The group of SLOCC can thus be described as the product of $G_{1,2,3}$ and the subgroup of $G_{8,15}$ that has trivial local action on the spatial mode of A. Therefore, the method to find the algebra of invariant polynomials of $G_{1,2,3}$ by Cayley's Omega process and then select the subalgebra invariant under the full group of SLOCC can be used similarly to the unconstrained case.

The lowest possible degree of a SLOCC invariant is four as in the unconstrained case.
For this degree, there is a unique SLOCC invariant 

\begin{align}\label{kry}
I_{A}^{(1)}=&2m_{\uparrow 0\diamondsuit}m_{\uparrow\diamondsuit 0}(m_{\downarrow\downarrow\uparrow}m_{\downarrow\uparrow\downarrow}-m_{\downarrow\downarrow\downarrow}m_{\downarrow\uparrow\uparrow})\nonumber\\
&+2m_{\downarrow 0\diamondsuit}m_{\downarrow\diamondsuit 0}(m_{\uparrow\downarrow\uparrow}m_{\uparrow\uparrow\downarrow}-m_{\uparrow\downarrow\downarrow}m_{\uparrow\uparrow\uparrow})\nonumber\\
&+m_{\uparrow 0\diamondsuit}m_{\downarrow\diamondsuit 0}(m_{\downarrow\downarrow\downarrow}m_{\uparrow\uparrow\uparrow}-m_{\downarrow\uparrow\downarrow}m_{\uparrow\downarrow\uparrow})\nonumber\\
&+m_{\uparrow 0\diamondsuit}m_{\downarrow\diamondsuit 0}(m_{\downarrow\uparrow\uparrow}m_{\uparrow\downarrow\downarrow}-m_{\downarrow\downarrow\uparrow}m_{\uparrow\uparrow\downarrow})\nonumber\\
&+m_{\downarrow 0\diamondsuit}m_{\uparrow\diamondsuit 0}(m_{\downarrow\downarrow\downarrow}m_{\uparrow\uparrow\uparrow}-m_{\downarrow\uparrow\downarrow}m_{\uparrow\downarrow\uparrow})\nonumber\\
&+m_{\downarrow 0\diamondsuit}m_{\uparrow\diamondsuit 0}(m_{\downarrow\uparrow\uparrow}m_{\uparrow\downarrow\downarrow}-m_{\downarrow\downarrow\uparrow}m_{\uparrow\uparrow\downarrow}).
\end{align}
For degree eight there is a single generator invariant. This can be chosen as a product of the three-tangle \cite{coffman} and $(m_{\uparrow 0\diamondsuit}m_{\downarrow\diamondsuit 0}-m_{\downarrow 0\diamondsuit}m_{\uparrow\diamondsuit 0})^2$, i.e., 

\begin{align}
I_{A}^{(2)}\!\!=&(m_{\uparrow 0\diamondsuit}m_{\downarrow\diamondsuit 0}-m_{\downarrow 0\diamondsuit}m_{\uparrow\diamondsuit 0})^2\nonumber\\
&\times\!(m_{\downarrow\downarrow\downarrow}^2m_{\uparrow\uparrow\uparrow}^2
+m_{\uparrow\downarrow\downarrow}^2m_{\downarrow\uparrow\uparrow}^2+
m_{\downarrow\uparrow\downarrow}^2m_{\uparrow\downarrow\uparrow}^2
+m_{\downarrow\downarrow\uparrow}^2m_{\uparrow\uparrow\downarrow}^2\nonumber\\
&-2m_{\downarrow\downarrow\downarrow}m_{\downarrow\downarrow\uparrow}
m_{\uparrow\uparrow\downarrow}m_{\uparrow\uparrow\uparrow}
-2m_{\downarrow\downarrow\downarrow}m_{\downarrow\uparrow\downarrow}
m_{\uparrow\downarrow\uparrow}m_{\uparrow\uparrow\uparrow}\nonumber\\
&-2m_{\downarrow\downarrow\downarrow}m_{\downarrow\uparrow\uparrow}
m_{\uparrow\downarrow\downarrow}m_{\uparrow\uparrow\uparrow}
-2m_{\downarrow\downarrow\uparrow}m_{\downarrow\uparrow\downarrow}
m_{\uparrow\downarrow\uparrow}m_{\uparrow\uparrow\downarrow}\nonumber\\
&-2m_{\downarrow\downarrow\uparrow}m_{\uparrow\uparrow\downarrow}
m_{\downarrow\uparrow\uparrow}m_{\uparrow\downarrow\downarrow}
-2m_{\downarrow\uparrow\downarrow}m_{\downarrow\uparrow\uparrow}
m_{\uparrow\downarrow\uparrow}m_{\uparrow\downarrow\downarrow}\nonumber\\
&+4m_{\downarrow\downarrow\downarrow}m_{\downarrow\uparrow\uparrow}
m_{\uparrow\downarrow\uparrow}m_{\uparrow\uparrow\downarrow}
+4m_{\uparrow\uparrow\uparrow}m_{\uparrow\downarrow\downarrow}
m_{\downarrow\uparrow\downarrow}m_{\downarrow\downarrow\uparrow}).\nonumber\\
\end{align}
The construction of these invariants is given in the Appendix.
For the cases with a fermion localized with party B or C invariants can be constructed in a completely analogous way.

No systematic search for this type of invariants was undertaken for higher degrees due to the increasing complexity of the search with increasing polynomial degree.
However, the fact that at least two independent invariant exist implies that there is an uncountable number of SLOCC inter-convertibility classes.

Finally, we can consider the case where the interaction between the two delocalized particles is strongly attractive. Then the local Hilbert spaces of B and C are both spanned by only $|0\rangle$ and $|\diamondsuit\rangle$. A SLOCC invariant for this case is

\begin{eqnarray}
I_{A}^{L}=(m_{\uparrow 0\diamondsuit}m_{\downarrow\diamondsuit 0}-m_{\downarrow 0\diamondsuit}m_{\uparrow\diamondsuit 0}).
\end{eqnarray}
If $I_{A}^{L}$ is nonzero the state features Bell-nonlocal correlations over the bipartition A|BC but not any other partitioning. Similar measures for systems with a particle localized with B or C can be constructed by permuting the spatial modes.

\section{Example: Groundstate entanglement at a transition between different regimes in a cyclic three site Ising-Hubbard chain}
\label{sec5}
In this section we consider a model system with three spin-$\frac{1}{2}$ fermions where the groundstate of the Hamiltonian undergoes a transition between two qualitatively different regimes.
Close to the transition the groundstate has a high degree of entanglement of the type for which only the invariants $I^{(1)}$ and $I^{(2)}$ are nonzero.

The three spin-$\frac{1}{2}$ fermions are confined to a three site cyclic chain. The Hamiltonian of the system is chosen as a hybrid of the Quantum Ising Hamiltonian \cite{sachdev} and the Hubbard Hamiltonian \cite{essler}. The potential energy terms are an on-site interaction between the fermions and a nearest neighbour Ising spin-interaction in the z-direction. Note that since the chain has a cyclic boundary condition all the sites are nearest neighbours.
In addition it features an external magnetic field aligned with the z-axis, i.e., the axis of the Ising interaction. The kinetic energy terms of the Hamiltonian correspond to spin reversal along the z-axis and inter-site hopping. The Hamiltonian $H$ is thus given by

 \begin{align}\label{ham}
H=&-J\sum_{j}\sigma^z_j\sigma^z_{j+1}-B\sum_j\sigma^z_j-K\sum_jn_{ j,s}n_{ j,-s}\nonumber\\&
+f\sum_j\sigma^x_j+\sum_{j}\sum_{s\in\{\downarrow,\uparrow\}}p_s(c^{\dagger}_{j,s}c_{j+1,s}+c^{\dagger}_{j+1,s}c_{j,s}),\nonumber\\
\end{align} 
where $c_{j,s}$, $c_{j,s}^{\dagger}$ and $n_{j,s}$ are the annihilation, creation, and number operators, respectively, of a fermion with spin $s\in\{\uparrow,\downarrow\}$ on site $j$. Here the spin $-s$ is understood to be the opposite spin of $s$. 
The parameter $J$ determines the type and strength of the Ising interaction. A positive $J$ corresponds to ferromagnetic interaction and a negative $J$ to anti-ferromagnetic interaction. The value of $K$ determines the type and strength of on-site interaction. A positive value of $K$ corresponds to attractive interaction and a negative value to repulsive interaction. The value of $B$ describes the external magnetic field.
The parameter $f$ determines the magnitude of the kinetic energy term describing the tendency for spins to reverse. The parameters $p_{\uparrow}$ and $p_{\downarrow}$ describe the inter-site hopping, i.e., the tendency for particles to move between adjacent sites. If $p_{\uparrow}\neq p_{\downarrow}$ the hopping is spin dependent. 

Here we choose the hopping parameters as $p=p_{\downarrow}=-p_{\uparrow}$. This choice of spin antisymmetry of the hopping terms is made to ensure that the ground level of the Hamiltonian is non-degenerate and a unique groundstate can be identified. 
Furthermore, we consider the case with a ferromagnetic Ising interaction and an attractive on-site fermion interaction.

If the on-site interaction dominates over the Ising interaction and the external field, and if the kinetic energy terms are much smaller than the potential energy terms, the groundstate is, to a good approximation, a superposition of different configurations where two fermions occupy the same site. We call this the paired regime.

If the Ising interaction dominates over the on-site interaction, and the kinetic energy terms are small, the groundstate of the Hamiltonian is close to a state where one fermion occupies each site. We call this the Ising regime.

The groundstate in the paired regime for $B=0$ and for $B\neq 0$ is not significantly different if $B\ll f$. This is because no potential barrier prevents the spins from reversing.
The groundstate in the Ising regime on the other hand is greatly affected by even a small external field. If $B=0$ the groundstate has equal amplitudes of $|\!\downarrow\downarrow\downarrow\rangle$
and $|\!\uparrow\uparrow\uparrow\rangle$.
In the presence of a small external field $B\gtrsim 1\cdot10^{-7}$ the groundstate is close to a product state where all the spins are aligned with the external field.
This is because the energy barrier corresponding to reversing a spin is large compared to the kinetic energy. Even though the energy difference $6B$ between the two product states  $|\!\downarrow\downarrow\downarrow\rangle$
and $|\!\uparrow\uparrow\uparrow\rangle$ is much smaller than $f$ by assumption, the energy needed to reverse a single spin to reach an intermediate state, e.g. $|\!\downarrow\uparrow\downarrow\rangle$, is approximately $4J$, which is much larger than $f$.

If the values of $J$ and $K$ are appropriately chosen, an increase of the external magnetic field $B$ can bring the groundstate of the system from the paired regime to the Ising regime.

To see this effect in the particular example Hamiltonian considered here the potential energy terms are chosen as $J=1$ and $K=2.99507$ and the kinetic terms are $f=5\cdot10^{-3}$ and $p=5\cdot10^{-6}$. The parameters $J$ and $K$ have been chosen to allow a very small external field to cause a transition between the Ising regime and the paired regime. This requires $K/J\approx 3$ so that the Ising interaction gives almost the same energy as the on-site interaction. The particular value of $K/J=2.99507$  has been carefully chosen such that the external field creates close to the maximal amount of groundstate entanglement possible with the given interaction terms. 
For these values of $J,K,f$ and $p$ the transition between the paired regime and the Ising regime is relatively abrupt around a value of 
$B=1.71\cdot 10^{-5}$. 

The energy levels and energy eigenstates of the Hamiltonian have been calculated numerically over a range of values of $B$ by the Mathematica 10.0 Eigensystem function \cite{mathematica} which uses the ARPACK software that is based on the Implicitly Restarted Arnoldi Method (IRAM). The numerical result indicates that for $B=0$ the lowest energy level is two-fold degenerate. But for $B>0$ the lowest level in non-degenerate and the groundstate is unique. Moreover, the Hamiltonian has an avoided levelcrossing between the groundstate and the first excited state close to $B=1.71\cdot 10^{-5}$. See Fig. \ref{fig:levelcross} for a plot of the energy of the groundstate and the four lowest excited energy levels for $B$ in the interval $B=0$ to $B=2\cdot 10^{-5}$. The energy is given relative to the groundstate energy at $B=0$.
The energy of the (initially) fourth excited energy level decreases as $B$ increases and
has an avoided levelcrossing with the groundstate at $B\approx 1.71 \cdot 10^{-5}$
after
crossing the third, second, and first excited levels.
This  avoided levelcrossing  corresponds to the transition between the paired regime and the Ising regime.

\begin{figure}
    \centering
        \includegraphics[width=0.46\textwidth]{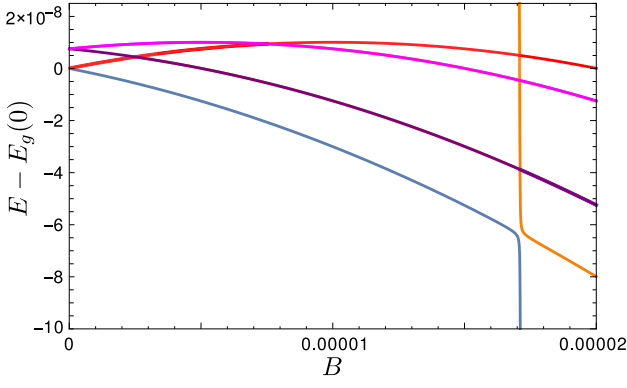}

    \caption{The energy $E$ of the groundstate (blue) and the four lowest excited levels (purple, magenta, red, orange), as functions of the external field $B$ in the interval $0\leq B\leq 2\cdot 10^{-5}$. 
    The energy is given relative to $E_g(0)$, the groundstate energy at $B=0$. 
The groundstate has an avoided levelcrossing with a non-degenerate excited level (orange) close to $B\approx 1.71\cdot 10^{-5}$.       }\label{fig:levelcross}
\end{figure}

For  $0<B\ll1.7\cdot 10^{-5}$ the groundstate of the system is very close to the state $1/\sqrt{12}|\psi_{P}\rangle$, where

\begin{align}|\psi_{P}\rangle\equiv&
-(|\diamondsuit  \!\uparrow\! 0\rangle+|\!\uparrow\! 0 \diamondsuit  \rangle+| 0\diamondsuit \!\uparrow\rangle)\nonumber\\
&+(| \!\uparrow\!\diamondsuit 0 \rangle+\!|0\!\uparrow \!\diamondsuit \rangle+|\diamondsuit 0 \!\uparrow  \rangle)
\nonumber\\
& -(|\diamondsuit 0 \!\downarrow \rangle+|\! \downarrow \!\diamondsuit 0 \rangle+|0\!\downarrow\!  \diamondsuit \rangle)\nonumber\\
&+(| \diamondsuit \! \downarrow \!  0\rangle+| 0\diamondsuit\! \downarrow \rangle+|\!\downarrow\!  0 \diamondsuit\rangle). 
\end{align}

Since the external field $B$ is very small compared to the kinetic term $f$ for the entire considered interval there is almost no directional bias of the spins in the groundstate as long as it is in the paired regime. Note that the state $1/\sqrt{12}|\psi_{P}\rangle$ is Bell-local for every partition of the system.

At $B=1.65\cdot 10^{-5}$ the groundstate has not deviated much from $1/\sqrt{12}|\psi_{P}\rangle$ and is

\begin{align}&2.20 \cdot 10^{-4}|\!\uparrow\uparrow\uparrow \rangle-1.20 \cdot 10^{-2}|\!\downarrow\downarrow\downarrow \rangle+0.288|\psi_{P}\rangle\nonumber\\
&-1.76\cdot 10^{-6}(|\!\downarrow\uparrow\uparrow \rangle+|\uparrow\downarrow\uparrow \rangle+|\!\uparrow\uparrow\downarrow \rangle)\nonumber\\
&+1.64\cdot 10^{-5}(|\!\downarrow\downarrow\uparrow \rangle+|\!\uparrow\downarrow\downarrow \rangle+|\!\downarrow\uparrow\downarrow \rangle).\end{align}
As $B$ increases beyond this value the amplitude of $|\!\downarrow\downarrow\downarrow\rangle$ begins to increase exponentially and reaches $\approx0.5$ at $B= 1.17099\cdot 10^{-5}$. At this point the groundstate is
\begin{align}&4.11 \cdot 10^{-4}|\!\uparrow\uparrow\uparrow \rangle-0.500|\!\downarrow\downarrow\downarrow\rangle +0.250|\psi_{P}\rangle\nonumber\\
&-3.33\cdot 10^{-6}(|\!\downarrow\uparrow\uparrow \rangle+|\!\uparrow\downarrow\uparrow \rangle+|\!\uparrow\uparrow\downarrow \rangle)\nonumber\\
&+6.27\cdot 10^{-4}(|\!\downarrow\downarrow\uparrow \rangle+|\!\uparrow\downarrow\downarrow \rangle+|\!\downarrow\uparrow\downarrow \rangle),\end{align}
which is close to $1/2|\!\downarrow\downarrow\downarrow \rangle+1/4|\psi_{P}\rangle$.
As $B$ increases beyond $1.171\cdot 10^{-5}$ the amplitude of $|\!\downarrow\downarrow\downarrow \rangle$ increases rapidly to almost 1 while the amplitude of $|\psi_{P}\rangle$ approaches zero.
At $B=1.175\cdot 10^{-5}$ the groundstate is already close to the product state $|\!\downarrow\downarrow\downarrow\rangle$,

\begin{align}&4.50\cdot 10^{-4}|\!\uparrow\uparrow\uparrow \rangle-0.999|\!\downarrow\downarrow\downarrow \rangle+5.36\cdot 10^{-3}|\psi_{P}\rangle\nonumber\\
&-3.71\cdot 10^{-6}(|\!\downarrow\uparrow\uparrow \rangle+|\!\uparrow\downarrow\uparrow \rangle+|\!\uparrow\uparrow\downarrow \rangle)\nonumber\\&+1.25\cdot 10^{-3}(|\!\downarrow\downarrow\uparrow \rangle+|\!\uparrow\downarrow\downarrow \rangle+|\!\downarrow\uparrow\downarrow \rangle).\end{align}

To quantify the entanglement of the groundstate as a function of the external field $B$, we consider the measures constructed from the generators of the polynomial $G_{{\mbox{\tiny\itshape SLOCC}}}$ invariants. 
The only two generators that take a non-negligible value for the groundstate are  $I^{(1)}$ and $I^{(2)}$. Due to the form of $|\psi_{P}\rangle$ we can see that $I^{(1)}\approx -I^{(2)}$ throughout the considered range of $B$. 
Therefore, we construct a measure as $4|I^{(1)}-I^{(2)}|^{1/2}$. The factor $4$ is a normalization chosen to give the measure the maximal value $1$.

For comparison we also consider the entanglement in the sector where one fermion is at each site. For this type of entanglement we use the measure $2|\tau|^{1/2}$ where $\tau$ is the three-tangle \cite{coffman}. The factor $2$ is chosen such that the maximal value of the measure is $1$.

In addition to these two measures based on polynomial invariants we also consider the subsystem entropy $\mathcal{E}$ of a site. Due to the form of the groundstate the subsystem entropy is the same for all three sites. The three different measures have a qualitatively different behaviour as the external field is increased from $B=0$ to $B=3\cdot 10^{-5}$.
 
The measure $4|I^{(1)}-I^{(2)}|^{1/2}$ is initially zero and then grows slowly until the growth becomes exponential close to $B=1.7\cdot 10^{-5}$ with a peak of $0.498$ at $B\approx 1.7099 \cdot 10^{-5}$. The value then rapidly goes to almost zero as B increases further. In Fig. \ref{fig:animals} the measure $4|I^{(1)}-I^{(2)}|^{1/2}$ is plotted for $B$ in the interval $B=0$ to $B=3\cdot 10^{-5}$ and also for the interval $B=1.65\cdot 10^{-5}$ to $B=1.75\cdot 10^{-5}$ around the transition.

\begin{figure}
    \centering
    \begin{subfigure}[b]{0.46\textwidth}
        \includegraphics[width=\textwidth]{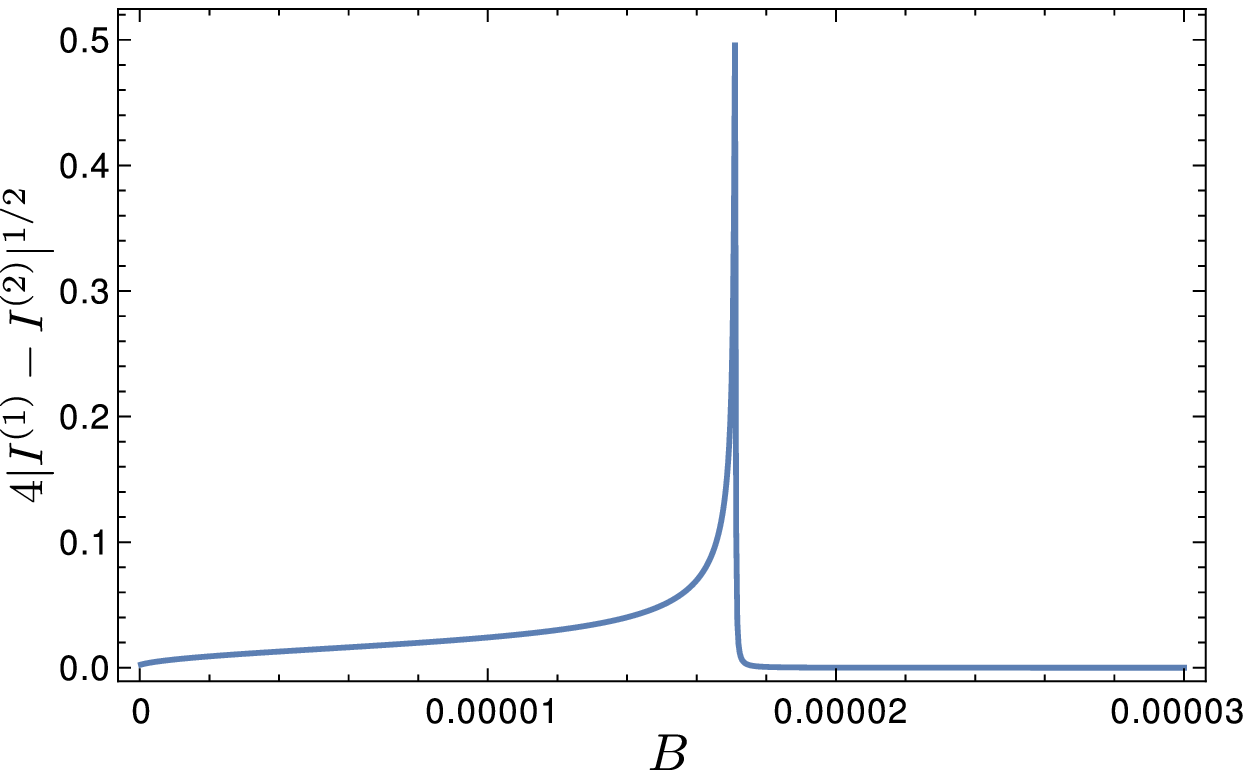}

    \end{subfigure}
    
    \begin{subfigure}[b]{0.46\textwidth}
        \includegraphics[width=\textwidth]{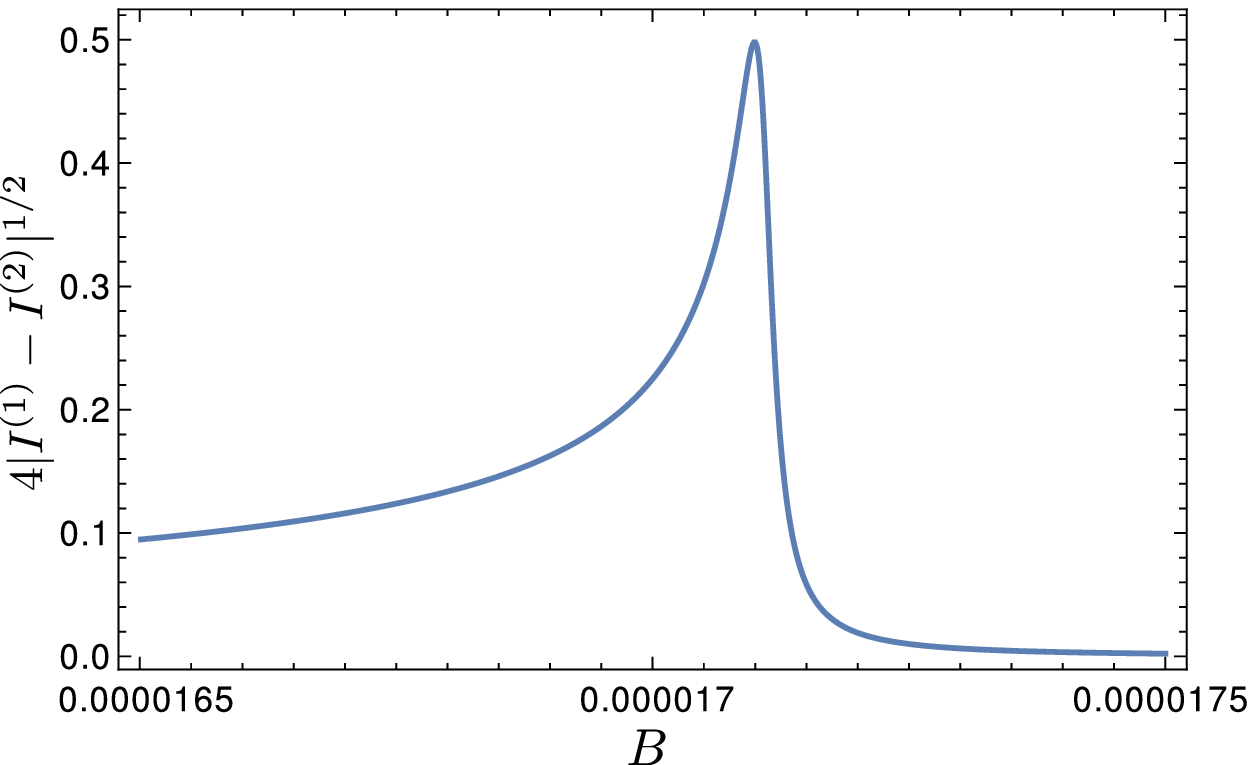}

    \end{subfigure}
    ~ 
    \caption{The measure $4|I^{(1)}-I^{(2)}|^{1/2}$ as a function of $B$ in the interval $B=0$ to $B=3\cdot 10^{-5}$ (upper) and the interval $B=1.65\cdot 10^{-5}$ to $B=1.75\cdot 10^{-5}$ (lower). The peak value 0.498 of $4|I^{(1)}-I^{(2)}|^{1/2}$ is at $B\approx 1.7099 \cdot 10^{-5}$. }\label{fig:animals}
\end{figure}

The measure $2|\tau|^{1/2}$ is close to zero at $B=0$ and stays close to zero until B approaches
$1.7\cdot 10^{-5}$ where it grows fast and peaks at a value $9.578\cdot 10^{-4}$ at $B\approx 1.7139 \cdot 10^{-5}$ . It then gradually decreases. In Fig. \ref{fig:animals2} the measure $2|\tau|^{1/2}$ is plotted for $B$ in the interval $B=0$ to $B=3\cdot 10^{-5}$ and also for the interval $B=1.65\cdot 10^{-5}$ to $B=1.75\cdot 10^{-5}$ around the transition. 

\begin{figure}
    \centering
    \begin{subfigure}[b]{0.46\textwidth}
        \includegraphics[width=\textwidth]{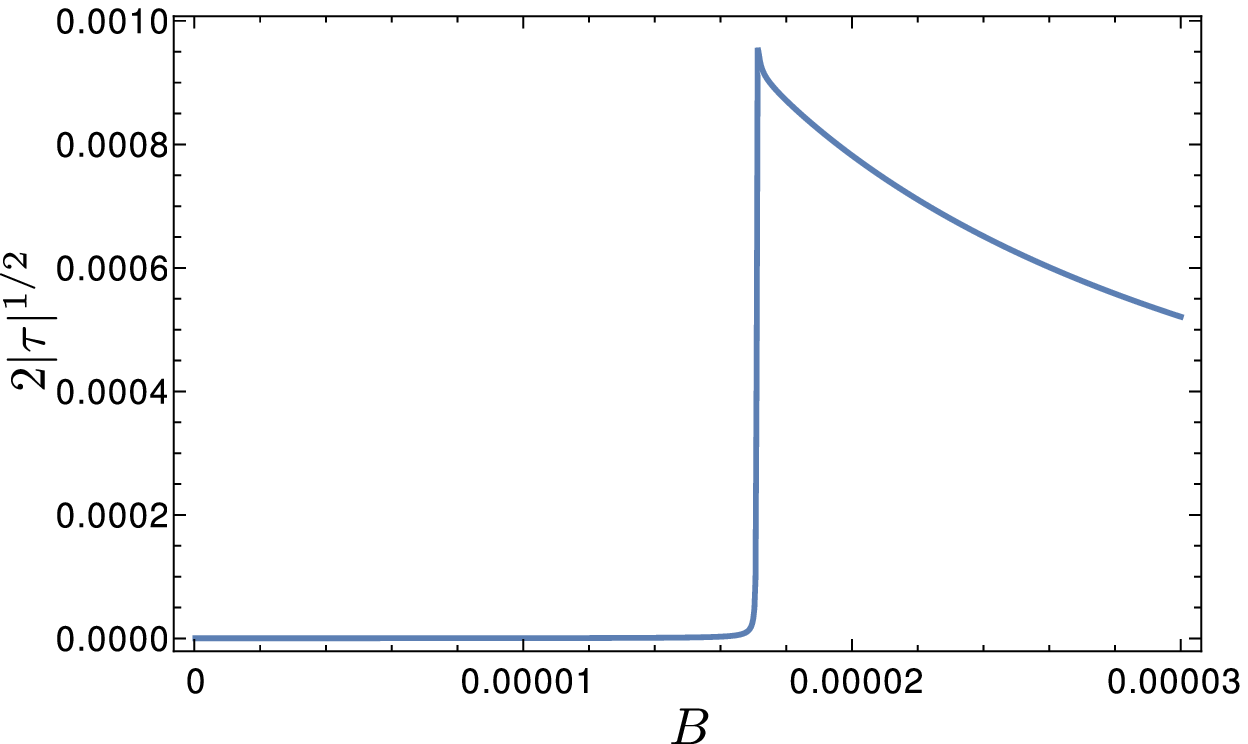}

    \end{subfigure}
    ~ 
    \begin{subfigure}[b]{0.46\textwidth}
        \includegraphics[width=\textwidth]{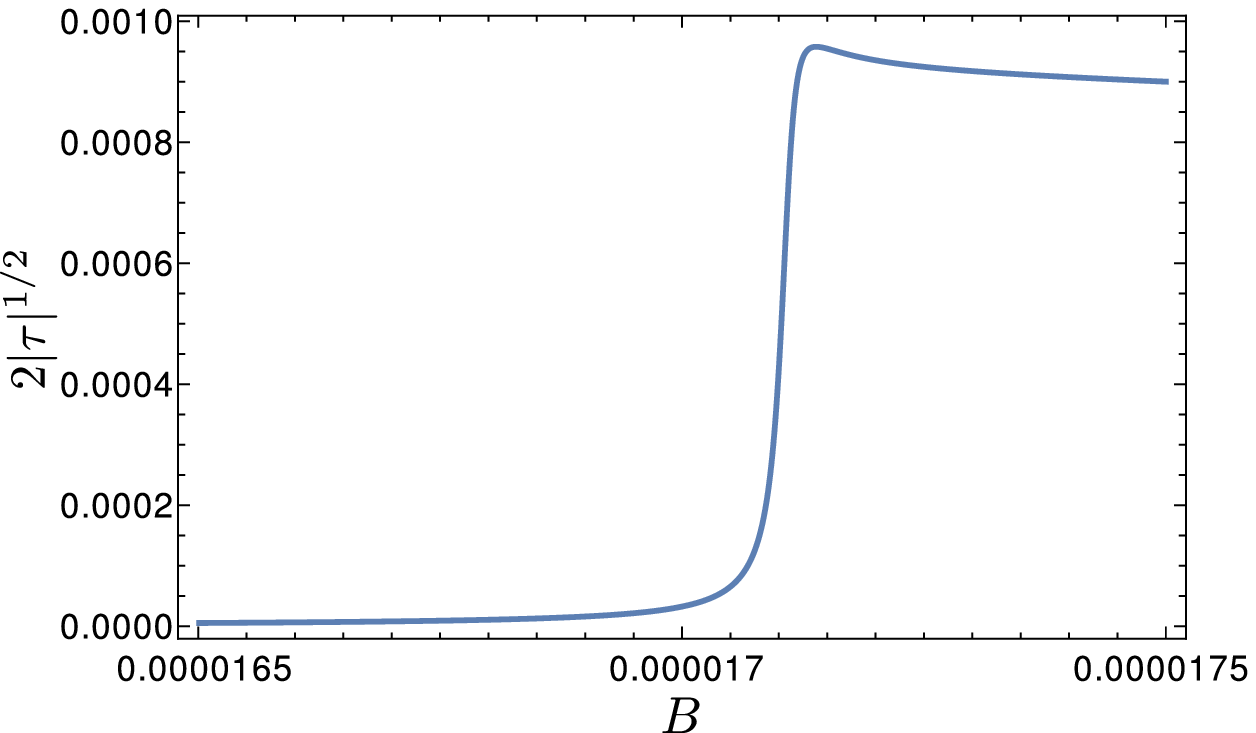}

    \end{subfigure}
    ~ 
    \caption{The measure $2|\tau|^{1/2}$ as a function of $B$ in the interval $B=0$ to $B=3\cdot 10^{-5}$ (upper) and the interval $B=1.65\cdot 10^{-5}$ to $B=1.75\cdot 10^{-5}$ (lower). The peak value $9.578\cdot 10^{-4}$ of $2|\tau|^{1/2}$ is at $B\approx 1.7139 \cdot 10^{-5}$. }\label{fig:animals2}
\end{figure}

The subsystem entropy has a value very close to $\ln(3)\approx 1.0986$ for the interval between $B=0$ and a value of $B$ close to $1.7\cdot 10^{-5}$. It then grows rapidly to a peak value of $\mathcal{E}=1.251$ at $B\approx 1.7097 \cdot 10^{-5}$ and afterwards rapidly goes to zero. In Fig. \ref{fig:animals3} the subsystem entropy is plotted for $B$ in the interval $B=0$ to $B=3\cdot 10^{-5}$ and also for the interval $B=1.65\cdot 10^{-5}$ to $B=1.75\cdot 10^{-5}$ around the transition.

\begin{figure}
    \centering
    \begin{subfigure}[b]{0.46\textwidth}
        \includegraphics[width=\textwidth]{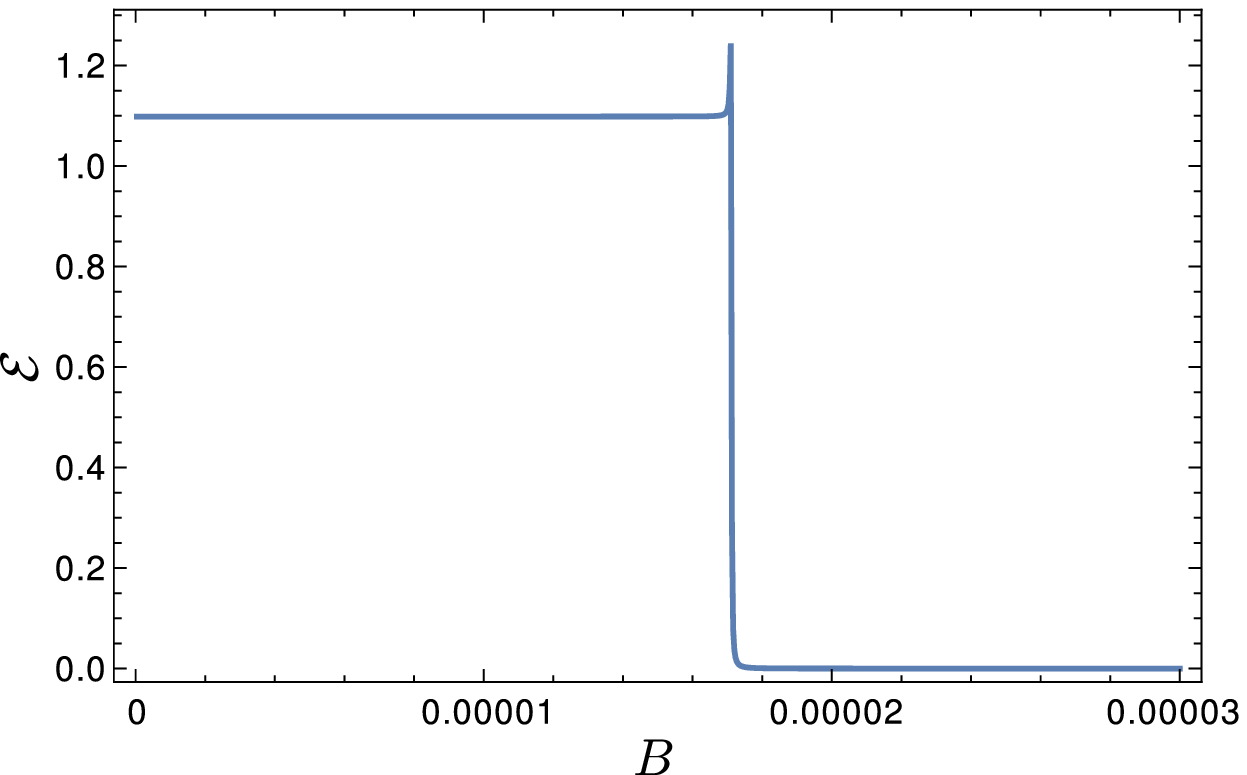}

    \end{subfigure}
    
    \begin{subfigure}[b]{0.46\textwidth}
        \includegraphics[width=\textwidth]{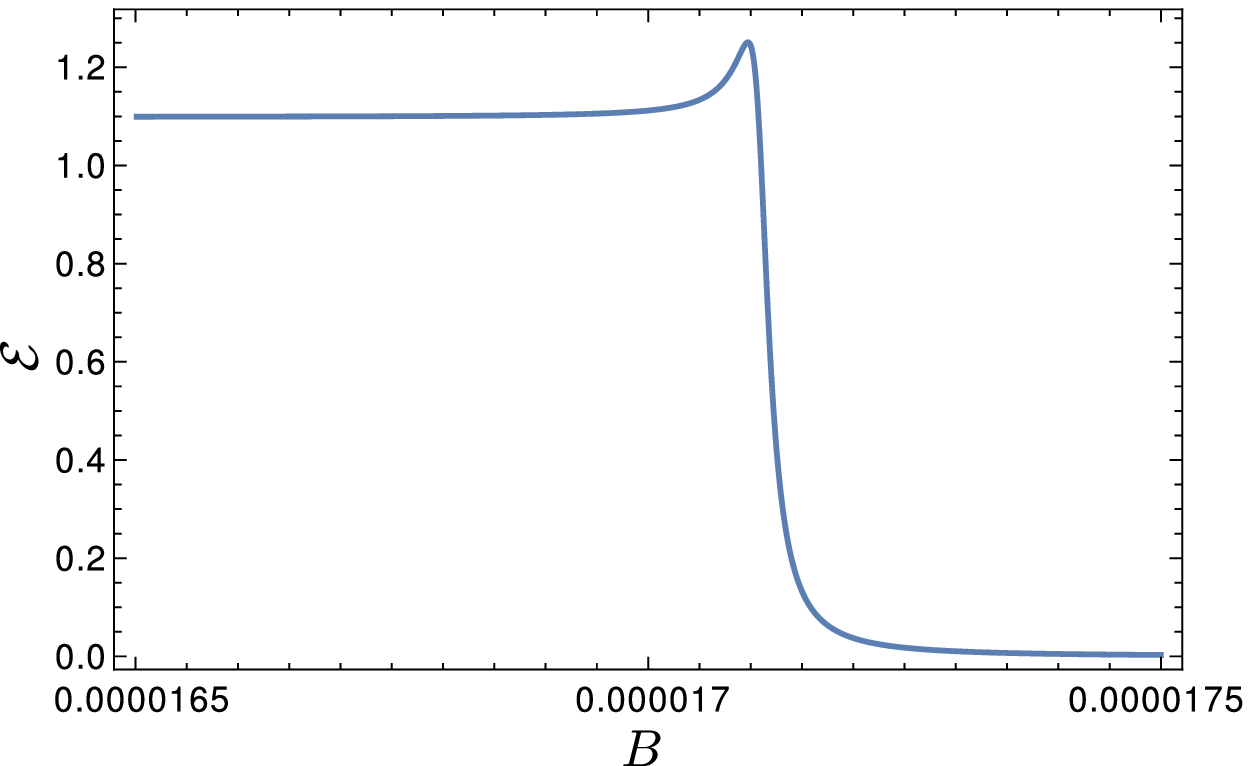}

    \end{subfigure}
   
    \caption{The subsystem entropy $\mathcal{E}$ as a function of $B$ in the interval $B=0$ to $B=3\cdot 10^{-5}$ (upper) and the interval $B=1.65\cdot 10^{-5}$ to $B=1.75\cdot 10^{-5}$ (lower). The peak value $1.251$ of the entropy is at $B\approx 1.7097 \cdot 10^{-5}$. }\label{fig:animals3}
\end{figure}

Any of the three measures can thus be used to detect the transition between the paired regime and the Ising regime. However, the measure $2|\tau|^{1/2}$ only reaches a small fraction of its maximal value since the entanglement in the sector where one fermion is at each site is very small.

The subsystem entropy reaches a value that indicates the presence of Bell-nonlocality, but the relative increase is small since this measure is sensitive also to the correlation of the groundstate in the paired regime. These correlations do not deviate much from Bell-local correlations since the groundstate in the paired regime, away from the transition, is close to indistinguishable from a separable state.

The measure $4|I^{(1)}-I^{(2)}|^{1/2}$ shows a large relative increase since it is only sensitive to correlations that are Bell-nonlocal. It also reaches a large value compared to its maximum since it was specifically designed to quantify this type of entanglement.
Note that for permutation symmetric or antisymmetric states maximal entanglement cannot be reached and the largest possible value of $4|I^{(1)}-I^{(2)}|^{1/2}$ is $2^{-1/4}\approx 0.841$. Furthermore given that the amplitude for $|\diamondsuit 0 \!\downarrow \rangle$ has the same absolute value as $|\diamondsuit 0 \!\uparrow \rangle$ and similarly for all permutations of the sites, the maximum value of $4|I^{(1)}-I^{(2)}|^{1/2}$ is $1/2$. Thus, the value of $4|I^{(1)}-I^{(2)}|^{1/2}$ at the transition is very close to the maximum possible given the properties of the Hamiltonian.

Without the big difference between the amplitudes of $|\!\downarrow\downarrow\downarrow\rangle$
and $|\!\uparrow\uparrow\uparrow\rangle$ in the groundstate at the transition the entanglement quantified by $4|I^{(1)}-I^{(2)}|^{1/2}$ would be very small. To see a high degree of this type of entanglement it is thus crucial that the kinetic energy $f$ is much smaller than the potential barrier between $|\!\downarrow\downarrow\downarrow\rangle$
and $|\!\uparrow\uparrow\uparrow\rangle$ in the Ising regime.

\section{Conclusions}

We have described a method for finding invariants under SLOCC in a system of indistinguishable spin-$\frac{1}{2}$ fermions delocalized over a number of spatial modes, with the constraint that the global number of particles is conserved. For two and three fermions such invariants and their associated measures were constructed. It was shown that if global particle conservation is the only constraint on the system a nonzero value of a SLOCC invariant implies that Bell-nonlocal correlations can be observed across any bipartition of the system into a single spatial mode and the rest of the modes. Furthermore it was shown that SLOCC invariants can only exist if the number of spin-$\frac{1}{2}$ fermions equal the number of spatial modes.
The result relating the existence of SLOCC invariants and Bell-nonlocality, as well as the condition on the number of fermions per spatial mode, were generalized to an arbitrary number of fermions of arbitrary spin.

For the two fermion case the properties of the constructed invariant under SLOCC and the associated measure of entanglement was compared to two other commonly used correaltion measures, the fermionic concurrence and the subsystem entropy. In particular their different relations to Bell-nonlocality was discussed.

The case with additional constraints, such as fixing the location one of the fermions or imposing a strong repulsive or attractive interaction was discussed. In this case the group of SLOCC depends on the extra constraints and therefore the set of invariants is different compared to the unconstrained case. The relation between Bell-nonlocality and SLOCC invariants which holds in the unconstrained case is not generally true when additional constraints are introduced. However, for the case of partial localization, i.e., when a subset of the fermions are in fixed spatial modes, as well as for the case of strong repulsive interaction, generalized relations exist. 

Finally, a model system with a hybrid Ising-Hubbard interaction Hamiltonian was considered to illustrate how the groundstate entanglement at a transition between a regime dominated by on-site interaction and a regime dominated by Ising interaction can be described by an entanglement measure constructed from a SLOCC invariant.
The value of the entanglement measure peaks at the transition and its behaviour was compared to the behaviour of a measure constructed from the three-tangle as well as the subsystem entropy.

A limitation of the constructed entanglement measures is that they work only for pure states. Measures valid also for mixed states would have a greater applicability as they could be used to describe entanglement in subsystems of larger more complex systems. Finding the convex roof extensions of the measures constructed here is thus a relevant open problem.

\subsection*{Acknowledgement}
The authors acknowledge discussions with Micha\l{} Oszmaniec, Antonio Ac\'\i n, Daniel Cavalcanti, Zolt\'an Zimbor\'as, Barbara Kraus, G\'eza T\'oth, and Manabendra Nath Bera.  
Support from the ERC CoG QITBOX, the Spanish MINECO (Project FOQUS FIS2013-46768-P, Severo Ochoa grant SEV-2015-0522), Fundacion Cellex, the Generalitat de Catalunya (SGR 875) and the John Templeton Foundation is acknowledged.
M.J. acknowledges support from  the Marie Curie COFUND action through the ICFOnest program.

\appendix
\section*{Appendix}
\label{app}
All invariants of the group $G_{1,2,3}$ can be constructed through Cayley's Omega Process. The Cayley Omega Process is designed to find the invariants of the special linear group and since $G_{1,2,3}$ for three parties is isomorphic to ${\mathrm{SL(2)}}^{\times 3}$ we can use it to find the $G_{1,2,3}$ invariants. 
For a detailed and general description of Cayleys Omega Process see e.g. Ref. \cite{olver} and see also Ref. \cite{ltt} for an application to invariants of unitary operations on a system of four distinguishable two level particles.

To use the Omega Process we must first represent the state of the system as a set of multilinear binary forms.
We arrange the amplitudes $m_{ijk}$ of the state vector for which $i,j,k\in\{\uparrow,\downarrow\}$ as a trilinear form $M(x_i,y_j,z_k)=\sum_{i,j,k\in\{\uparrow,\downarrow\}}m_{ijk}x_iy_jz_k$, and the remaining amplitudes are arranged in six linear forms

\begin{align}
m_{21}(x_i)&=\sum_{i\in\{\uparrow,\downarrow\}}m_{i\diamondsuit 0 }x_i,\nonumber\\
m_{31}(x_i)&=\sum_{i\in\{\uparrow,\downarrow\}}m_{i 0\diamondsuit}x_i,\nonumber\\
m_{12}(y_j)&=\sum_{j\in\{\uparrow,\downarrow\}}m_{\diamondsuit j 0}y_j,\nonumber\\
m_{32}(y_j)&=\sum_{j\in\{\uparrow,\downarrow\}}m_{0 j\diamondsuit}y_j,\nonumber\\
m_{23}(z_k)&=\sum_{k\in\{\uparrow,\downarrow\}}m_{0 \diamondsuit k}z_k,\nonumber\\
m_{13}(z_k)&=\sum_{k\in\{\uparrow,\downarrow\}}m_{\diamondsuit 0 k }z_k.
\end{align}

The invariants are obtained recursively from the collection $\{M,m_{12},m_{13},m_{21},m_{23},m_{31},m_{32}\}$ through iterating a type of operation called transvection defined by the following two steps. In the first step multiply two forms, e.g, $A(x'_i,y'_j,z'_k)$ and $B(x''_i,y''_j,z''_k)$.
Then, apply a partial differential operator $\Omega_w$, where $w\in\{x,y,z\}$ to $A(x'_i,y'_j,z'_k)B(x''_i,y''_j,z''_k)$. The operator $\Omega_x$ is defined by

\begin{eqnarray}\Omega_{x}=\frac{\partial^2}{\partial{x'_{\uparrow}}\partial{x''_{\downarrow}}}-\frac{\partial^2}{\partial{x''_{\uparrow}}\partial{x'_{\downarrow}}},\end{eqnarray}
and $\Omega_y$ and $\Omega_z$ are defined analogously. 
This is followed by the substitution of $x$ for $x'$ and $x''$. If several operators $\Omega_i$ are applied, all of the involved variables are substituted. For example the application of $\Omega_x\Omega_z$ is followed by the substitution $x$ for $x'$ and $x''$ and $z$ for $z'$ and $z''$.

If the result of a transvection is a new form, further transvections can be performed involving this form. If the result of a sequence of transvections is a scalar, this scalar is invariant under $G_{1,2,3}$.

As a shortform notation we denote a transvection of two forms $A$ and $B$ on one index, e.g. $i$, by $A_{ijk}B_{ilm}$, a transvection on
two indices, e.g. $i,j$, by $A_{ijk}B_{ijl}$, and so on. Note that $A_{ijk}B_{ilm}$ is a quadrilinear form with indices $j,k,l,m$ and $A_{ijk}B_{ijl}$ is a bilinear form with indices $k,l$.

From the $G_{1,2,3}$-invariants that can be constructed by Cayley's Omega Process  we select the subset which is also invariant under $G_{8,15}$. This is the subset where each monomial term in the polynomial has an equal number of indices of the four possible ones  $\downarrow$, $\uparrow$, $0$, and $\diamondsuit$. For example,
the monomial $m_{0\diamondsuit\uparrow}m_{\diamondsuit 0\downarrow}
m_{0\uparrow\diamondsuit}m_{\diamondsuit\downarrow 0}
m_{\uparrow 0\diamondsuit}m_{\uparrow\diamondsuit 0}m_{\downarrow\downarrow\uparrow}m_{\downarrow\uparrow\downarrow}$ is invariant under $G_{8,15}$ while $m_{\uparrow 0\diamondsuit}m_{\uparrow\diamondsuit 0}m_{\downarrow\downarrow\uparrow}m_{\downarrow\uparrow\downarrow}$ is not.
The requirement of invariance under $G_{8,15}$ restricts the degrees of the invariant polynomials to be multiples of four.

All the possible transvections of the collection $\{M,m_{12},m_{13},m_{21},m_{23},m_{31},m_{32}\}$ which are also invariant under $G_{8,15}$ were constructed for degrees 4, 8 and 12. A set of seven invariants from which all other invariants of these degrees can be constructed through multiplication and addition was selected. None of the invariants in this set can be constructed from the other elements of the set. This was seen by testing for linear dependence between all polynomials of degree 4, 8 or 12, respectively, that can be constructed by multiplication of the seven chosen invariants. 

A search for invariants of degree 16 did not yield any invariants that could not be constructed from the lower degree invariants. No search was conducted for higher degrees.

Of the seven invariants chosen as generators, the two degree four invariants $I_1$ and $I_2$ are given by

\begin{align}
I^{(1)}&=M_{ijk}m_{31i}m_{12j}m_{23k},\nonumber\\
I^{(2)}&=M_{ijk}m_{21i}m_{32j}m_{13k}.
\end{align}
The three degree eight invariants $I_{BC},I_{AC}$ and $I_{AB}$ are given by 
\begin{align}
I_{BC}&=m_{21i}M_{ijk}M_{ljk}m_{31l}m_{12n}m_{32n}m_{13p}m_{23p},\nonumber\\
I_{AC}&=m_{32j}M_{ijk}M_{ilk}m_{12l}m_{21n}m_{31n}m_{13p}m_{23p},\nonumber\\
I_{AB}&=m_{23k}M_{ijk}M_{ijl}m_{13l}m_{12n}m_{32n}m_{21p}m_{31p}.
\end{align}
Finally the two degree twelve invariants $I_{ABC}^{(1)}$ and $I_{ABC}^{(2)}$ are

\begin{widetext}
\begin{align}
I_{ABC}^{(1)}&=m_{23k}M_{ijk}M_{ijl}M_{npl}m_{31n}m_{12p}m_{31q}m_{21q}m_{32r}m_{12r}m_{23s}m_{13s},\nonumber\\
I_{ABC}^{(2)}&=m_{13k}M_{ijk}M_{ijl}M_{npl}m_{21n}m_{32p}m_{31q}m_{21q}m_{32r}m_{12r}m_{23s}m_{13s}
.\end{align}
\end{widetext}

\subsubsection{Partial localization}

In the case of the partial localization described in Sect. \ref{partloc} all states of the system can be represented by the subset
$\{M,m_{21},m_{31}\}$ of the forms. The unique degree four polynomial SLOCC invariant in Eq. \ref{kry} can be constructed from these forms as

\begin{eqnarray}
I_{A}^{(1)}=m_{21i}M_{ijk}M_{ljk}m_{31l}.
\end{eqnarray}
The degree eight invariant can be constructed as

\begin{eqnarray}
I_{A}^{(2)}=m_{21i}M_{ijk}M_{ljk}M_{lnp}M_{qnp}m_{31q}m_{21r}m_{31r},
\end{eqnarray}
or alternatively as
\begin{eqnarray}
I_{A}^{(2)}=\frac{1}{2}M_{ijk}M_{ljk}M_{lnp}M_{inp}m_{21q}m_{31q}m_{21r}m_{31r},
\end{eqnarray}
where $\frac{1}{2}M_{ijk}M_{ljk}M_{lnp}M_{inp}$ is the three-tangle \cite{coffman}, i.e., the hyperdeterminant of $M$. 
Analogously, for the cases where a fermion is localized with B or C, invariants can be constructed from the subsets $\{M,m_{12},m_{32}\}$ and $\{M,m_{13},m_{23}\}$, respectively.

When one fermion is localized with party A and the interaction of the two delocalized fermions is constrained to be strongly attractive the possible states of the system are described by the subset
$\{m_{21},m_{31}\}$, and all SLOCC invariants are generated by the polynomial

\begin{eqnarray}
I_{A}^{L}=m_{21r}m_{31r}.
\end{eqnarray}
Analogously, if the localized fermion is in mode B or C, a single invariant can be constructed from the subset $\{m_{12},m_{32}\}$ or $\{m_{13},m_{23}\}$, respectively.

\end{document}